\documentclass[12p]{article}
\usepackage{amsmath,amssymb,amsthm}
\usepackage{hyperref}
\usepackage{multirow}
\usepackage{times}
\usepackage{dsfont}
\usepackage{hhline}
\usepackage{caption}	
\usepackage[left=2.8cm, right=2.8cm, top=2.8cm, bottom=3.2cm]{geometry}
\usepackage{fullpage}
\usepackage{cite}
\usepackage{graphicx}
\usepackage{tikz}
\usetikzlibrary{arrows}
\newtheorem{theorem}{Theorem}[section]
\newtheorem{lemma}[theorem]{Lemma}

\newtheorem{definition}[theorem]{Definition}
\theoremstyle{remark}
\newtheorem{remark}[theorem]{Remark}

\newcommand{\real}{\mathbb{R}}
\newcommand{\complex}{\mathbb{C}}

\newcommand{\ripl}{\mathrm{RIP}_L}
\newcommand{\etasm}{\eta_{\mathbf{s,M}}}

\newcommand{\res}[1]{\ensuremath{#1\!\times\!#1}}
\newcommand{\smsparse}{(\mathbf{s},\mathbf{M})}

\newcommand{\weightedsparse}{$(\mathbf{\omega},s)$-sparse }

\newcommand{\rp}{\text{Re}}

\usepackage{float}
\usepackage{array,blindtext}
\usepackage{pifont}% http://ctan.org/pkg/pifont
\newcommand{\cmark}{\ding{51}}%
\newcommand{\xmark}{\ding{55}}%

\newcount\colveccount
\newcommand*\colvec[1]{
        \global\colveccount#1
        \begin{pmatrix}
        \colvecnext
}
\def\colvecnext#1{
        #1
        \global\advance\colveccount-1
        \ifnum\colveccount>0
                \\
                \expandafter\colvecnext
        \else
                \end{pmatrix}
        \fi
}

\numberwithin{equation}{section}
\date{}
\graphicspath{ %
{images/} %
{images/newSkepsImages/} %
{images/pngImages/} %
{images/pngImages/Matrices/} %
}

\title{On the absence of the RIP in real-world applications of compressed sensing and the RIP in levels}
\author{Alexander Bastounis \footnotemark[1] \and Anders C. Hansen \footnotemark[2]}

\begin{document}
\maketitle

\renewcommand{\thefootnote}{\fnsymbol{footnote}}

\footnotetext[1]{CCA, Centre for Mathematical Sciences, University of Cambridge, UK (A.Bastounis@maths.cam.ac.uk)}
\footnotetext[2]{DAMTP, Centre for Mathematical Sciences, University of Cambridge, UK (A.Hansen@damtp.cam.ac.uk)}
\begin{abstract}

The purpose of this paper is twofold. The first is to point out that the Restricted Isometry Property (RIP) does not hold in many applications where compressed sensing is successfully used. This includes fields like Magnetic Resonance Imaging (MRI), Computerized Tomography, Electron Microscopy, Radio Interferometry and Fluorescence Microscopy. We demonstrate that for natural compressed sensing matrices involving a level based reconstruction basis (e.g. wavelets), the number of measurements required to recover all $s$-sparse signals for reasonable $s$ is excessive. In particular, uniform recovery of all $s$-sparse signals is quite unrealistic. This realisation shows that the RIP is insufficient for explaining the success of compressed sensing in various practical applications. The second purpose of the paper is to introduce a new framework based on a generalised RIP-like definition that fits the applications where compressed sensing is used. We show that the shortcomings that show that uniform recovery is unreasonable no longer apply if we instead ask for structured recovery that is uniform only within each of the levels. To examine this phenomenon, a new tool, termed the 'Restricted Isometry Property in Levels' is described and analysed. Furthermore, we show that with certain conditions on the Restricted Isometry Property in Levels, a form of uniform recovery within each level is possible. Finally, we conclude the paper by providing examples that demonstrate the optimality of the results obtained.
\end{abstract}
\section{Introduction}
Compressed Sensing (CS), introduced by Cand\`es, Romberg \& Tao 
\cite{CandesRombergTao} and Donoho \cite{donohoCS}, has been one of the important new developments in applied mathematics in the last decade \cite{candesCSMag,Cohen09compressedsensing, DonohoTannerCounting,EldarDuarteCSReview,boche2015compressed, EldarKutyniokCSBook,FornasierRauhutCS,foucartBook, RombergCompImg, HermanStrohmerRadar}. By introducing a non-linear reconstruction method via convex optimisation, one can circumvent traditional barriers for reconstructing vectors that are sparse or compressible, meaning that they have few non-zero coefficients or can be approximated well by vectors with few non-zero coefficients. 

A substantial part of the CS literature has been devoted to the Restricted Isometry Property (RIP) (see Definition \ref{RIP}). Matrices which satisfy the RIP of order $s$ (see Remark \ref{remark:RIPHolds}) can be used to perfectly recover all $s$-sparse vectors (i.e. vectors with at most $s$ non-zero coefficients) via $\ell^1$ minimisation. Remarkably, this is possible even if the matrix in question is singular.  

Given the substantial interest in the RIP over the last years it is natural to ask whether this intriguing mathematical concept is actually  observed in many of the applications where CS is applied. It is well known that verifying that the RIP holds for a general matrix is an NP hard problem \cite{ComputationalComplexityOfRIP}, however, there is a simple test that can be used to show that certain matrices do not satisfy the RIP. This is called the \emph{flip test}. As this test reveals, there are an overwhelming number of practical applications where one will not observe the RIP of order $s$ for any reasonable sizes of the sparsity $s$. In particular, this list of applications includes Magnetic Resonance Imaging (MRI), other areas of medical imaging such as Computerised Tomography (CT) and all the tomography cousins such as thermoacoustic, photoacoustic or electrical impedance tomography, electron 
microscopy, seismic tomography, as well as other fields such as fluorescence microscopy, Hadamard 
spectroscopy and radio interferometry. Typically, we shall see that practical applications that exploit sparsity in a level based reconstruction basis will not exhibit the RIP of order $s$ for reasonable $s$. 

We will thoroughly document the lack of RIP of order $s$ in this paper, and explain why it does not hold for reasonable $s$. It is then natural to ask whether there might be an alternative to the RIP that may be more suitable for the actual real world CS applications. With this in mind, we shall introduce the \emph{RIP in levels} which generalises the classical RIP and is much better suited for the actual applications where CS is used.

\subsection{Compressed sensing} 
We shall begin by discussing the general ideas of compressed sensing. Typically we are given a scanning device, represented by an invertible matrix $M \in \complex^{n \times n}$, and we seek to recover information $x \in \complex^n$ from observed measurements $y:=Mx$. In general, we require knowledge of every element of $y$ to be able to accurately recover $x$ without additional structure. Indeed, let $\Omega = \left\lbrace \omega_1,\omega_2, \dotsc, \omega_{|\Omega|} \right\rbrace $ with $1 \leq \omega_1 < \omega_2 < \omega_3 <  \dotsb < \omega_{|\Omega|} \leq n$ and define the projection map $P_{\Omega}: \complex^n \to \complex^{|\Omega|}$ so that $P_{\Omega}(x_1,x_2, \dotsc, x_n) := (x_{\omega_1},x_{\omega_2},\dotsc, x_{\omega_{|\Omega|}})$. If $|\Omega|$ is strictly less than $n$ then for a given $y$ there are at least two distinct vectors $x_1 \in \complex^n$ and $x_2 \in \complex^n$ with $P_{\Omega} y = P_{\Omega}Mx_1 = P_{\Omega}Mx_2$, so that knowledge of $P_{\Omega}y$ will not allow us to distinguish between multiple candidates for $x$. Ideally though we would like to be able to take $|\Omega| \ll n$ to reduce either the computational or financial costs associated with using the scanning device $M$.

So far, we have not assumed any additional structure on $x$. However, let us consider the case where the vector $x$ consists mostly of zeros. More precisely, we make the following definition:
\begin{definition}[Sparsity]\label{Sparsity}
A vector $x$ is said to be \emph{$s$-sparse} for some natural number $s$ if\\
$|\text{\emph{supp}}(x)|\leq s$,
where $\text{\emph{supp}}(x)$ denotes the support of $x$.
\end{definition}

The key to CS is the fact that, under certain conditions, finding solutions to the $\ell^1$ problem
\begin{equation}\label{eq:l1Min}
\min \|\widehat{x}\|_1 \text{ such that } P_{\Omega} M x = P_{\Omega}M\widehat{x}
\end{equation}
gives a good approximation to $x$. More specifically, in \cite{ripstart} Cand{\`e}s and Tao described the concept of \emph{Restricted Isometry Property}. Typically, $M$ is an isometry. With this in mind, it may seem plausible that $U:=P_{\Omega}M$ is also close to an isometry when acting on sparse vectors. Specifically, we make the following definition:
\begin{definition}[Restricted Isometry Property]\label{RIP}
For a fixed matrix $U \in \complex^{m \times n}$, the  \emph{Restricted Isometry Property} (RIP) constant of order $s$, denoted by $\delta_s$, is the minimal $\delta > 0$ such that
\begin{equation}\label{ripDefinition}
(1-\delta)\|x\|_2^2 \leq \|Ux\|^2_2 \leq (1+\delta) \|x\|^2_2
\end{equation}
for all $s$-sparse vectors $x \in \complex^n$. 
\end{definition}
A typical theorem in compressed sensing is similar to the following, proven in \cite{4over41Paper} and \cite{foucartBook}.
\begin{theorem}\label{Theorem:RIPtypeTheorem}
For natural numbers $m$ and $n$, let $y \in \complex^n$ and $|\Omega| = m$ so that $P_{\Omega}y \in \complex^m$. Let $U \in \complex^{m \times n}$.
Additionally, suppose that $U$ has RIP constant for $2s$-sparse vectors given by $\delta_{2s}$, where $\delta_{2s} < \frac{4}{\sqrt{41}}$. Then for any $x \in \complex^n$ such that
$
\|Ux  - P_{\Omega}y \|_2 \leq \epsilon   
$
and any $\tilde{x}$ which solve the following $\ell^1$ minimisation problem:
\begin{equation} \label{ripl1Problem}
\min \|\widehat{x}\|_1 \text{ subject to } \|U\widehat{x} - P_{\Omega}y\|_2 \leq \epsilon
\end{equation} 
we have
\begin{equation} \label{equation:conclusionOfRIPTypeTheorem}
\|x-\tilde{x}\|_1 \leq C \sigma_{s}(x)_1 + D\epsilon \sqrt{s}
\end{equation}
and
\begin{equation} \label{equation:l2conclusionOfRIPTypeTheorem}
\|x - \tilde{x}\|_2 \leq \frac{C}{\sqrt{s}} \sigma_{s}(x)_1 + D \epsilon
\end{equation}
where $
\sigma_{s}(x)_1 := \min \lbrace\|x-\widehat{x}_2\|_1 \text{ such that } \widehat{x}_2 \text{ is } s \text{ sparse}
\rbrace$
and $C,D$ depend only on $\delta_{2s}$.
\end{theorem}

Theorem \ref{Theorem:RIPtypeTheorem} roughly says that if the matrix $U$ has a sufficiently small RIP constant for $2s$-sparse vectors then \begin{enumerate}
\item If $x$ is $s$-sparse then we can recover it from knowledge of $Ux$ by solving the (convex) $\ell^1$ minimisation problem.
\item If we only know a noisy version of $Ux$ (i.e we are given $y = Ux + \eta$ for some noise $\eta$ satisfying $|\eta|< \epsilon$) then we can still recover an approximation to $x$ which will be accurate up to the size of the noise.
\item If $x$ is very close (in $\ell^1$ norm) to an $s$-sparse vector then we can recover an approximation to $x$ which will be accurate up to the distance between $x$ and its closest $s$-sparse vector.  
\end{enumerate}
For a fixed matrix $U$ (in particular, if $U = P_{\Omega}M$ then $U$ is fixed only if the set $\Omega$ is fixed), the ability to recover all $s$-sparse vectors $x$ using the values of $Ux$ is commonly known as \emph{uniform recovery}. We shall see in the remaining sections that expecting uniform recovery is unrealistic in a variety of situations.

\begin{remark}\label{remark:TheoremVariants}
Variants on Theorem \ref{Theorem:RIPtypeTheorem} include changing the hypothesis to $\delta_{2s} < C$ for a different constant $C$ (e.g. \cite{CandesWeakRIP} requires $\delta_{2s} < \sqrt{2} - 1$ or \cite{FOUCARTRIP4652} uses different methods to show that, for large s, $\delta_{2s}<0.475$ is sufficient) or changing $\delta_{2s}$ to $\delta_{3s}$ (e.g. \cite{delta3k} 
requires $\delta_{3s} < 4-2\sqrt{3}$) or $\delta_{s}$ (e.g. \cite{SharpRIPBounds}
requires $\delta_{s} < \frac{1}{3}$). We will see that any of these modifications will suffer from the same shortcomings in the next section. 
\end{remark}
\begin{remark}\label{remark:RIPHolds}
If the RIP constant for $U \in \complex^{m \times n}$ of order $as$ for some $a \in \mathbb{N}$ is sufficiently small so that the conclusion of a theorem similar to Theorem \ref{Theorem:RIPtypeTheorem} holds then we say that \emph{$U$ satisfies or exhibits the RIP of order $s$}.
\end{remark}
\section{The absence of the RIP and the flip test}

\subsection{The flip test}\label{Section:flipTest}
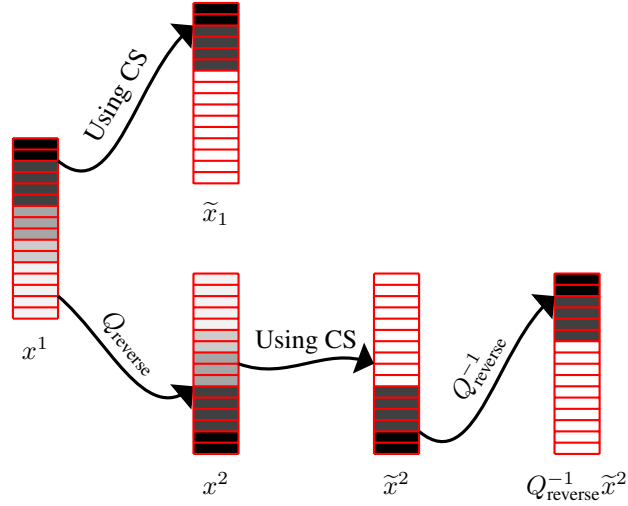
\begin{figure}[t]
\centering
\newcommand{\scaleSize}[0]{0.6}
\newcommand{\vectorLength}{4}
\newcommand{\xoff}{0}
\newcommand{\yoff}{0}
\newcommand*{\INVISIBLE}[1]{{\color{white}#1}}
%tikz picture here

\begin{tikzpicture}[scale=\scaleSize]
	 	 	 	 %initial vector
	 	 	 	 \fill[black!100] (1*1,4*1) rectangle (0*1,3.5*1);
	 	 	 	 \fill[black!75!white] (1*1,3.5*1) rectangle (0,2.5*1);
	 	 	 	 \fill[black!35!white] (1*1,2.5*1) rectangle (0,1.75*1);
	 	 	 	 \fill[black!20!white] (1*1,1.75*1) rectangle (0,1.25*1);
	 	 	 	 \fill[black!5!white] (1*1,1.25*1) rectangle (0,0*1);
	 	 	 	 \draw[red,thick] (0,0) grid [xstep=1,ystep=0.25*1] (1,4*1);
	 	 	 	 \renewcommand{\yoff}{0}
	 	 	 	 \newcommand{\yoffPrime}{-0.75*\vectorLength}
	 	 	 	 \node[anchor=north,align=center]  at
	 	 	    (0.5,-0.2) {\INVISIBLE{$w'_2$}$x^1$\INVISIBLE{$w'_2$}};
	 	 	 	 
	 	 	 	 %tildex1
	 	 	 	 \renewcommand{\xoff}{1*\vectorLength}
	 	 	 	 \renewcommand{\yoff}{0}
	 	 	 	 \renewcommand{\yoffPrime}{0.75*\vectorLength}
	 	 	 	 \draw[very thick,-triangle 90] (1+\xoff -\vectorLength,\vectorLength - 0.5+\yoff) to[out = -40, in = -140] node[pos=0.6,above ,rotate=60]{Using CS} (\xoff,3.5+\yoffPrime);
	 	 	 	 \renewcommand{\yoff}{\yoffPrime}
	 	 	 	 \fill[black!100] (\xoff+1,4+\yoff) rectangle (\xoff,3.5+\yoff);
	 	 	 	 	 	 	 	 	 	 		 	 	 	    	 	 	 	 	 	 \fill[black!75!white] (\xoff+1,3.5 + \yoff) rectangle (\xoff,2.5*1+ \yoff);
	 	 	 	 	 	 	 	 	 	 		 	 	 	    	 	 	 	 	 	 \fill[white] (\xoff + 1,2.5 + \yoff) rectangle (\xoff,1.75*1+ \yoff);
	 	 	 	 	 	 	 	 	 	 		 	 	 	    	 	 	 	 	 	 \fill[white] (\xoff + 1,1.75+ \yoff) rectangle (\xoff,1.25*1+ \yoff);
	 	 	 	 	 	 	 	 	 	 		 	 	 	    	 	 	 	 	 	 \fill[white] (\xoff + 1,1.25+ \yoff) rectangle (\xoff,0*1+ \yoff);
	 	 	 	 	 	 	 	 	 	 		 	 	 	    	 	 	 	 	 	 \draw[red,thick] (0+\xoff,0+ \yoff-0.01) grid [xstep=1,ystep=0.25*1] (1 + \xoff,4*1+ \yoff);
	 	 	 	  	 	 	    	 	 	 	 	 	   
	 	 	 	 	 	 	 	 	 	 		 	 	 	    	 	 	 	 	 	 \node[anchor=north,align=center]  at
	 	 	 	  	 	 	    	 	 	 	 	 	    (0.5+\xoff,-0.2+\yoff) {$\widetilde{x}_1$};
	 	 	    
	 	 	 	 %next vector
	 	 	 	 \renewcommand{\yoffPrime}{-0.75*\vectorLength}
	 	 	 	 \renewcommand{\xoff}{\vectorLength}
	 	 	 	 \renewcommand{\yoff}{0}
	 	 	 	 \draw[very thick,-triangle 90] (1+\xoff -\vectorLength,0.5+\yoff) to[out = -40, in = -140] node[pos=0.4,above ,rotate=-50]{$Q_{\text{reverse}}$} (\xoff,1.5+\yoffPrime);
 	 	 	 	 	\renewcommand{\yoff}{\yoffPrime}
	 	 	 	 	 	 \fill[black!100] (\xoff+1,0+\yoff) rectangle (\xoff,0.5+\yoff);
	 	 	 	 	 	 	 	 	 \fill[black!75!white] (\xoff+1,0.5 + \yoff) rectangle (\xoff,1.5*1+ \yoff);
	 	 	 	 	 	 	 	 	 \fill[black!35!white] (\xoff + 1,1.5 + \yoff) rectangle (\xoff,2.25*1+ \yoff);
	 	 	 	 	 	 	 	 	 \fill[black!20!white] (\xoff + 1,2.25+ \yoff) rectangle (\xoff,2.75*1+ \yoff);
	 	 	 	 	 	 	 	 	 \fill[black!5!white] (\xoff + 1,2.75+ \yoff) rectangle (\xoff,4+ \yoff);
	 	 	 	 	 	 	 	 	 \draw[red,thick] (0+\xoff,0+ \yoff) grid [xstep=1,ystep=0.25*1] (1 + \xoff,4*1+ \yoff);
	 	   
	 	 	 	 	 	 	 	 	 \node[anchor=north,align=center]  at
	 	    (0.5+\xoff,-0.2+\yoff) {\INVISIBLE{$w'_2$}$x^2$\INVISIBLE{$w'_2$}};
  	 	 	 	 	 	\renewcommand{\xoff}{2*\vectorLength}
		 	 	 	 \fill[black!100] (\xoff+1,0+\yoff) rectangle (\xoff,0.5+\yoff);
		 	 	 	 	 	 	 \fill[black!75!white] (\xoff+1,0.5 + \yoff) rectangle (\xoff,1.5*1+ \yoff);
		 	 	 	 	 	 	 \fill[white] (\xoff + 1,1.5+\yoff) rectangle (\xoff,\yoff+4);
		 	 	 	 	 	 	 \draw[red,thick] (0+\xoff,0+ \yoff) grid [xstep=1,ystep=0.25*1] (1 + \xoff,4*1+ \yoff);

		 	 	 	 	 	 	 \node[anchor=north,align=center]  at (0.5+\xoff,-0.2+\yoff) {\INVISIBLE{$w'_2$}$\widetilde{x}^2 	$\INVISIBLE{$w'_2$}};
 	 	 	 	 	 		 	 	 \draw[very thick,-triangle 90] (1+\xoff -\vectorLength,2+\yoff) to[out = -20, in = 160] node[anchor=south,above]{Using CS} (\xoff,2+\yoff);
  	 	 	 	 	 		 	 	  	 	 	 	 	 	\renewcommand{\xoff}{3*\vectorLength}
	 	 	 	 	 		 	 	 	    	 	 	 \fill[black!100] (\xoff+1,4+\yoff) rectangle (\xoff,3.5+\yoff);
	 	 	 	 	 		 	 	 	    	 	 	 	 	 	 \fill[black!75!white] (\xoff+1,3.5 + \yoff) rectangle (\xoff,2.5*1+ \yoff);
	 	 	 	 	 		 	 	 	    	 	 	 	 	 	 \fill[white] (\xoff + 1,2.5 + \yoff) rectangle (\xoff,1.75*1+ \yoff);
	 	 	 	 	 		 	 	 	    	 	 	 	 	 	 \fill[white] (\xoff + 1,1.75+ \yoff) rectangle (\xoff,1.25*1+ \yoff);
	 	 	 	 	 		 	 	 	    	 	 	 	 	 	 \fill[white] (\xoff + 1,1.25+ \yoff) rectangle (\xoff,0*1+ \yoff);
	 	 	 	 	 		 	 	 	    	 	 	 	 	 	 \draw[red,thick] (0+\xoff,0+ \yoff) grid [xstep=1,ystep=0.25*1] (1 + \xoff,4*1+ \yoff);
 	 	    	 	 	 	 	 	   
	 	 	 	 	 		 	 	 	    	 	 	 	 	 	 \node[anchor=north,align=center]  at
 	 	    	 	 	 	 	 	    (0.5+\xoff,-0.2+\yoff) {\INVISIBLE{$w'_2$}$Q^{-1}_{\text{reverse}}\widetilde{x}^2$\INVISIBLE{$w'_2$}};
 	 	    	 	 	 	 	 	    
  	 	 	    	 	 	 	 	 	    	 	 	 	 	 	 	 	 	 	\draw[very thick,-triangle 90] (1+\xoff -\vectorLength,0.5+\yoff) to[out = -40, in = -140] node[pos=0.6,above ,rotate=60]{$Q^{-1}_{\text{reverse}}$} (\xoff,3.5+\yoff);
  	 	 	    	 	 	 	 	 	     	 	 	 	 	 	\renewcommand{\xoff}{4*\vectorLength}

	 	 	 	 \end{tikzpicture}

\caption{A graphical demonstration of the flip test for matrices which exhibit the RIP. Darker colours denote larger values.}
\label{fig:flipTestDemo}
\end{figure}
Although Theorem \ref{Theorem:RIPtypeTheorem} and similar theorems based on the RIP seem convenient, computing the RIP constant of an arbitrary matrix is an NP hard problem \cite{ComputationalComplexityOfRIP}. Certain special cases have been shown to have a small enough RIP constant with high probability (e.g. \cite{SmallRIP} for Gaussian and Bernoulli matrices) for the conclusion of Theorem \ref{Theorem:RIPtypeTheorem} to hold 
 but in other cases the size of the RIP constant is not known. In \cite{FlipTest} the following test (the so called 'flip test') was proposed. The test can be used to verify the RIP of order $s$ for reasonable $s$ and uniform recovery of $s$-sparse vectors and works as follows:
\begin{enumerate}
\item Suppose that a matrix $U \in \complex^{m \times n}$ has a sufficiently small RIP constant so that the error estimate \eqref{equation:conclusionOfRIPTypeTheorem} holds. Take a vector $x^1$ and compute $y^1:=Ux^1$. Set $\widetilde{x}^1$ to be the solution to \eqref{ripl1Problem} where we seek to reconstruct $x^1$ from $y^1$. 
%\item Because the property of being $s$-sparse is independent of ordering, if \eqref{ripDefinition} holds for $U$ then it also holds for the matrix $UP$ where $P$ is a permutation matrix. 
\item For an operator $Q$ that permutes entries of vectors in $\complex^n$, set $x^2:=Qx^1$, and compute $y^2 := Ux^2$. Again, set $\widetilde{x}^2$ to be the solution to \eqref{ripl1Problem} where we seek to reconstruct $x^2$ from $y^2$.
\item From \eqref{equation:conclusionOfRIPTypeTheorem}, we have 
%\begin{equation*}
$
\|x^1-\widetilde{x}^1\|_1, \|x^2-\widetilde{x}^2\|_1 \leq C \sigma_{s}(x^2)_1 =  C \sigma_{s}(x^1)_1
$
%\end{equation*}
because $\sigma_s$ is independent of ordering of the coefficients. Also, since permutations are isometries,
$\|x^1-Q^{-1}\widetilde{x}^2\|_1 \leq C \sigma_{s}(x^2)_1 = C \sigma_{s}(x^1)_1.$
\item If $x^1$ is well approximated by $\widetilde{x}^1$ and the RIP is the reason for the good error term, we expect 
$\|x^1-\widetilde{x}^1\|_1 \approx C \sigma_{s}(x^1)_1 $. Thus,  $x^1$ should also be well approximated by $Q^{-1}\widetilde{x}^2$.
\item Hence if $\|x^1-\widetilde{x}^1\|_1 \,\, \text{differs greatly from}   \,\, \|x^1-Q^{-1}\widetilde{x}^2\|_1,$ there is no uniform recovery and we cannot have the RIP.
\end{enumerate}
The particular choice of $Q$ that was given in \cite{FlipTest} was the permutation $Q_{\text{reverse}}$ that reverses order - namely, if $x \in \complex^n$ then $Q_{\text{reverse}}(x_1,x_2 ,\dotsc ,x_{n-1},x_n) = (x_n,x_{n-1}, \dotsc, x_2,x_1).$
A graphical demonstration of the results of the flip test (assuming that the RIP holds for large enough $s$ to recover typical $x^1$ and using the permutation $Q_{\text{reverse}}$) is given in Figure \ref{fig:flipTestDemo}. Here, larger values are represented by darker colours, so that the first few values of $x^1$ are the largest. This is what we would expect if $x^1$ is the wavelet coefficients of a typical image. 

\begin{table}[t]

\begin{center}
\begin{tabular}{ |l|l|c|c|c| }
\hline
\multicolumn{2}{|c|}{}& \multicolumn{2}{ |c| }{Matrix method} & Observes RIP\\
\cline{3-4}
\multicolumn{2}{|c|}{} & $\mathrm{DFT}\cdot \mathrm{DWT^{-1}}$& $\mathrm{HAD}\cdot \mathrm{DWT^{-1}}$ & \\ \hline
\multirow{9}{*}{Problem} & MRI & \cmark & \xmark & \xmark \\ \cline{2-5}
 & Tomography & \cmark & \xmark & \xmark \\ \cline{2-5}
 
  & Spectroscopy   & \cmark & \xmark & \xmark \\ \cline{2-5}
  
    & Electron microscopy  & \cmark & \xmark & \xmark \\ \cline{2-5}
 & Radio interferometry   & \cmark & \xmark & \xmark \\ \cline{2-5}
  & Fluorescence microscopy & \xmark & \cmark & \xmark \\ \cline{2-5}
  & Lensless camera  & \xmark & \cmark & \xmark \\ \cline{2-5}
  & Single pixel camera  & \xmark & \cmark & \xmark \\ \cline{2-5}
  & Hadamard spectroscopy & \xmark & \cmark & \xmark \\ \cline{2-5}
\hline
\end{tabular}
\end{center}

\caption{A table displaying various applications of compressive sensing. For each application, a suitable matrix is suggested along with information on whether or not that matrix has a sufficiently small RIP constant for results similar to the conclusion of Theorem \ref{Theorem:RIPtypeTheorem} to hold.}\label{table:ObeysRIP}
\end{table}

This test is simple, and performing it with the permutation $Q_{\text{reverse}}$ on a variety of examples, even with the absence of noise (see Figure \ref{fig:flipTestImage}),
gives us a collection of problems for which the RIP does not account for the excellent reconstruction observed by experimental methods. The notation $\mathrm{DFT},\mathrm{HAD}$ and $\mathrm{DWT}$ is used throughout Figure \ref{fig:flipTestImage} and the remainder of this article to represent the Discrete Fourier Transform, the Hadamard Transform and the Discrete Wavelet Transform respectively. For several other examples with flip tests that verify the lack of RIP in many important compressed sensing applications, see \cite{SMSparsity} and \cite{AsymptoticStructure}. The conclusion of the flip test for compressed sensing is summarised in Table \ref{table:ObeysRIP}.  

It is worth noting that even with $97\%$ sampling as in the second row of Figure \ref{fig:flipTestImage}, the RIP constant $\delta_s$ (for reasonable $s$) of some matrices is still too large, so that either inequalities \eqref{equation:conclusionOfRIPTypeTheorem} and \eqref{equation:l2conclusionOfRIPTypeTheorem} do not hold or they are insufficiently tight to be useful statements about the quality of recovery obtained by $\ell^1$ minimisation. Although this may seem surprising at first, it is actually very natural when one understands how poor the recovery of the finer wavelet levels is for subsampling of matrices like $\mathrm{DFT}\cdot\mathrm{DWT}^{-1}_4$.
Note that the flip test will fail in a similar manner if we replace wavelets with other popular frames such as curvelets, contourlets or shearlets \cite{candes2004new, Gitta3, Vetterli}.
\begin{figure}
\begin{center}
\newcommand{\abrcap}[2]{\raisebox{80pt}{\begin{minipage}[t]{65pt}{\footnotesize
 #1x#1\\[15pt]Error:\\[2pt]}#2\%\end{minipage}}}%
\newcommand{\abrdesc}[6]{\raisebox{110pt}{\begin{minipage}[t]{60pt}{\footnotesize
				#6\\\res{#1}\\#2\%\\[4pt]#3$\cdot$#4\\[4pt]{#5}}\end{minipage}}}%
\newcommand{\abrflip}[1]{\includegraphics[width=0.28\textwidth]{#1}}%
\vspace{-10pt}
\begin{tabular}[h]{@{\hspace{0pt}}c@{\hspace{0pt}}c@{\hspace{3pt}}c@{\hspace{3pt}}c@{\hspace{0pt}}}
& CS reconstruction & CS reconstruction w/ flip & Subsampling pattern 
used\\[5pt]
  \abrdesc{2048}{12}{$\mathrm{DFT}$}{$\mathrm{DWT}^{-1}_3$}{MRI,\\Spectroscopy, Radio-\\interferometry}{College 1}
   &\abrflip{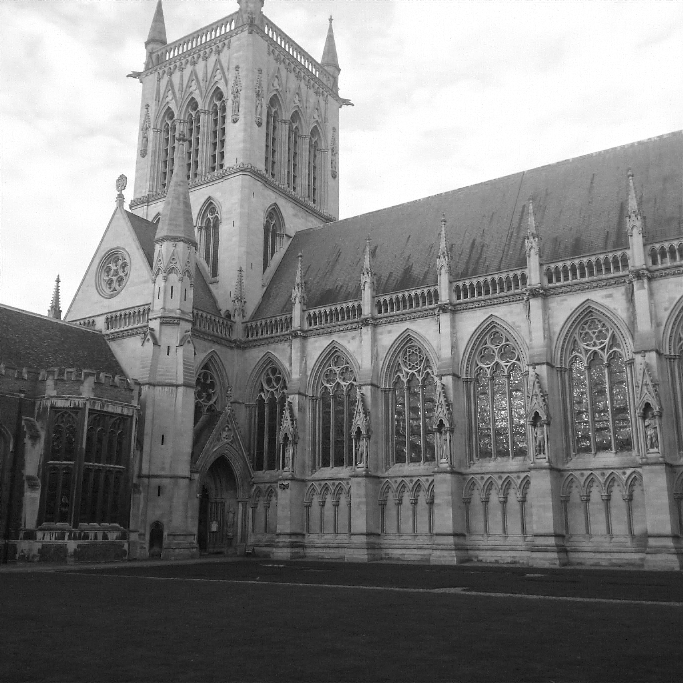}
  &\abrflip{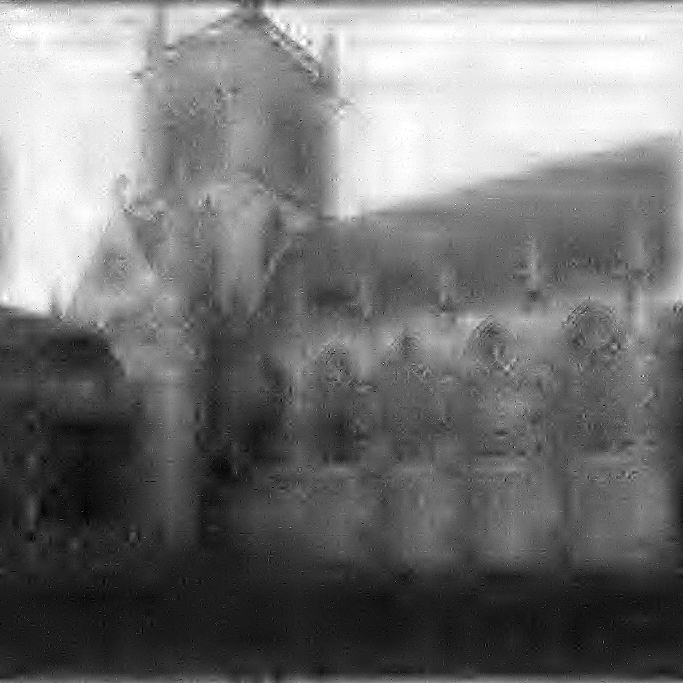}
  &\abrflip{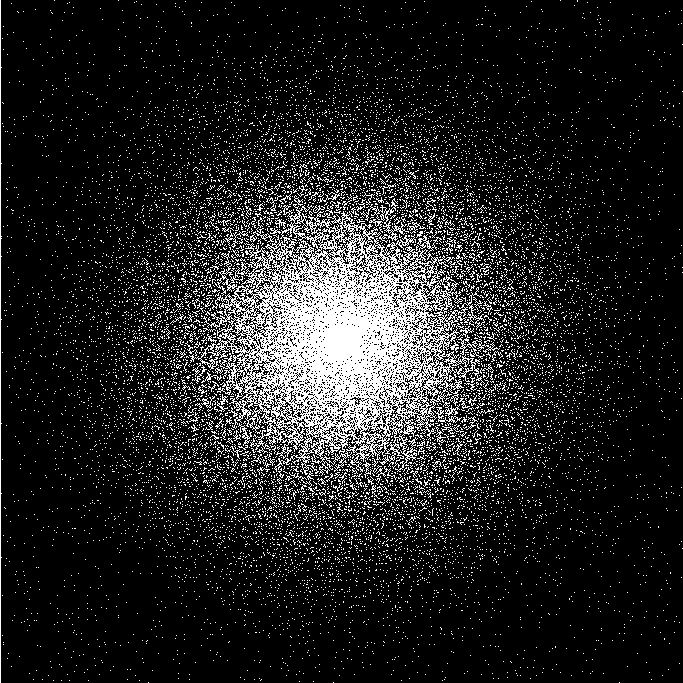}
   \\
     \abrdesc{2048}{97}{$\mathrm{DFT}$}{$\mathrm{DWT}^{-1}_4$}{MRI,\\Spectroscopy, Radio-\\interferometry}{College 2}
      &\abrflip{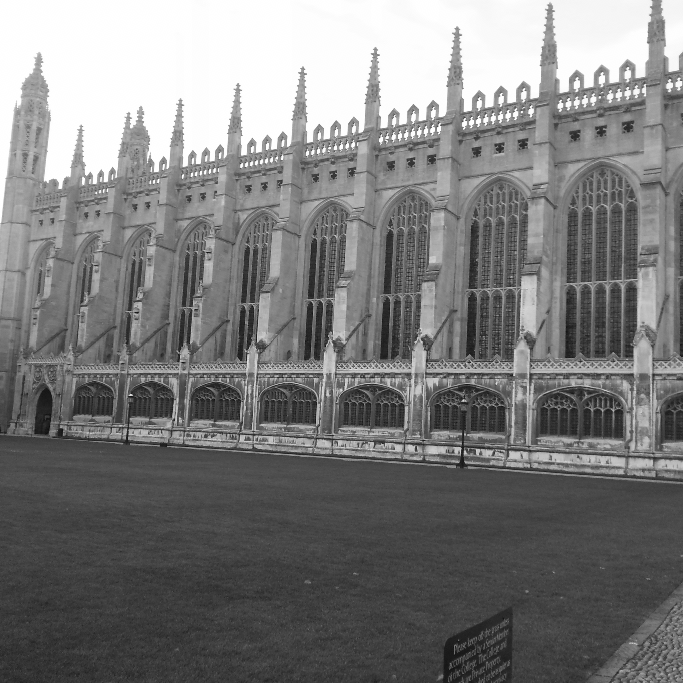}
        &\abrflip{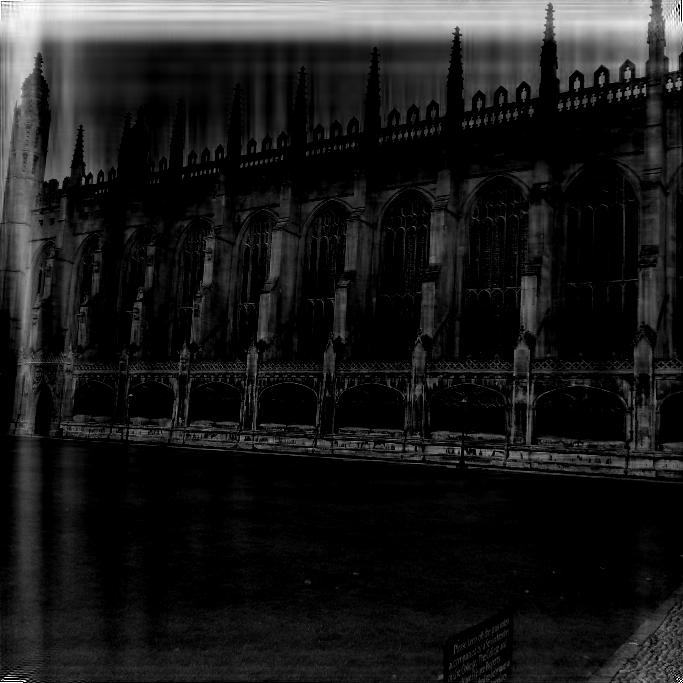}
        &\abrflip{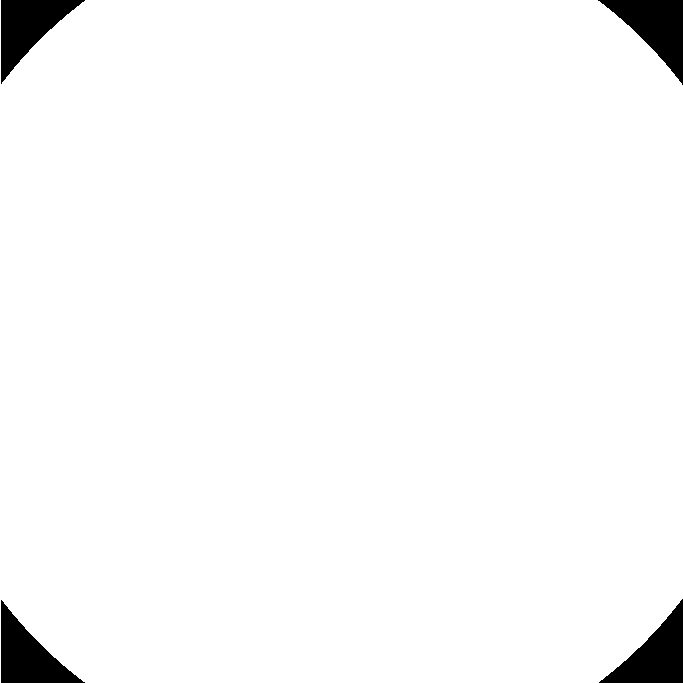}
    \\
       \abrdesc{2048}{21}{$\mathrm{HAD}$}{$\mathrm{DWT}^{-1}_2$}{Comp. imag.,\\ Hadamard spectroscopy, Fluorescence microscopy }{Rocks}
        &\abrflip{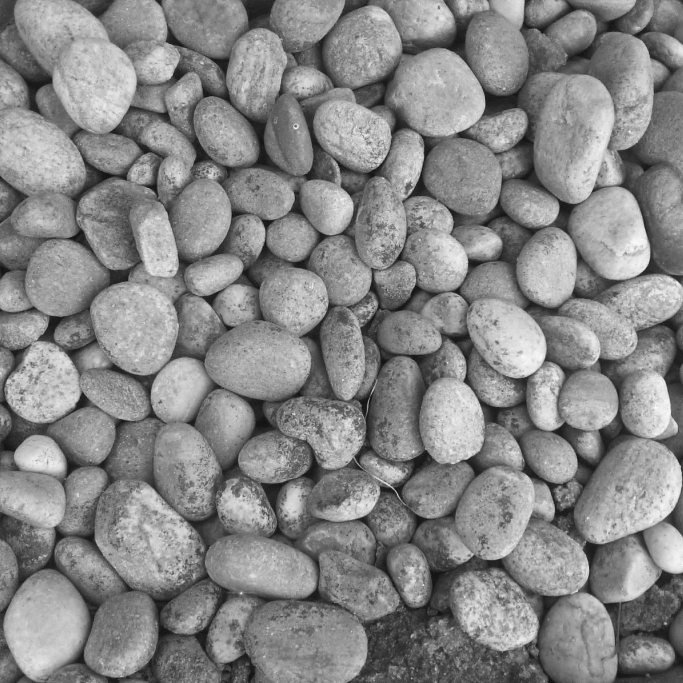}
          &\abrflip{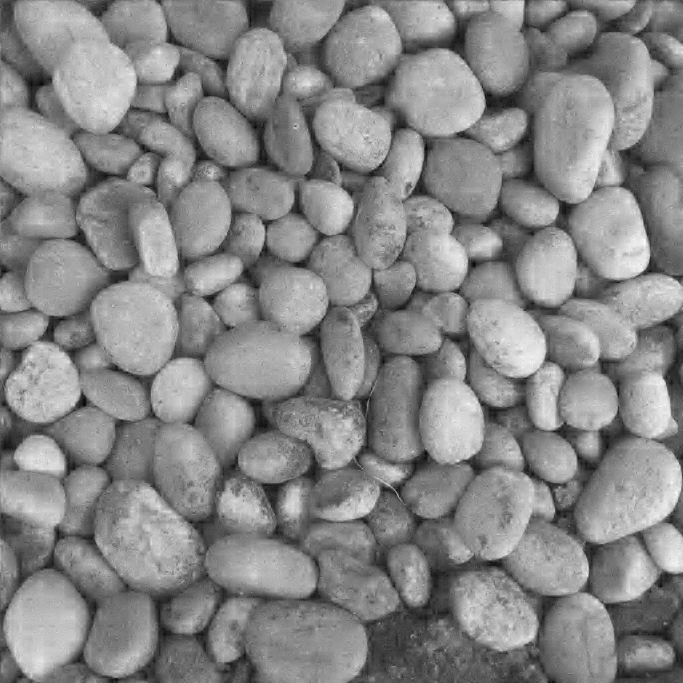}
          &\abrflip{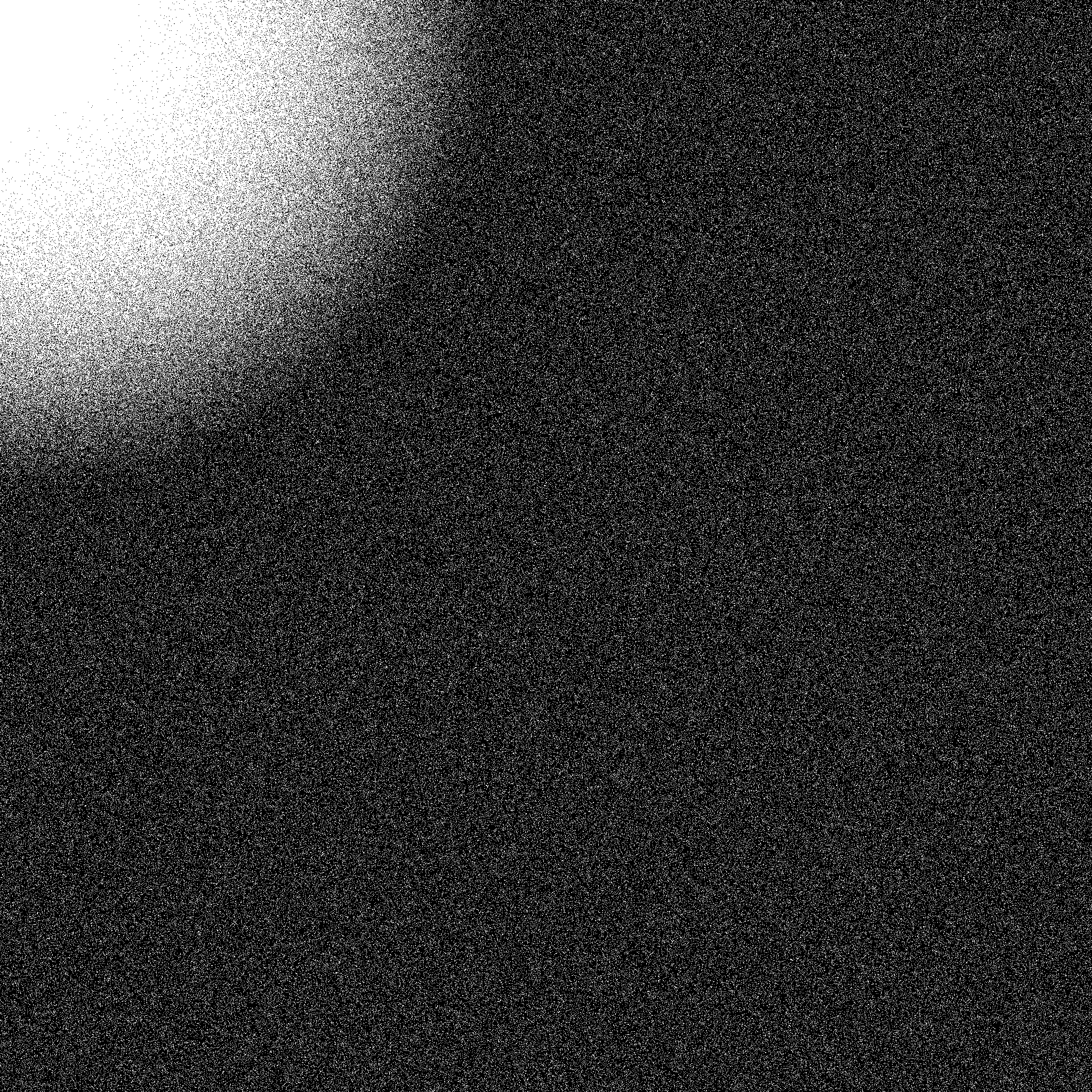}
	\\
       \abrdesc{2048}{27.5}{$\mathrm{DFT}$}{$\mathrm{DWT}^{-1}_2$}{Electron \\microscopy,\\ Computerised\\ tomography}{College3}
        &\abrflip{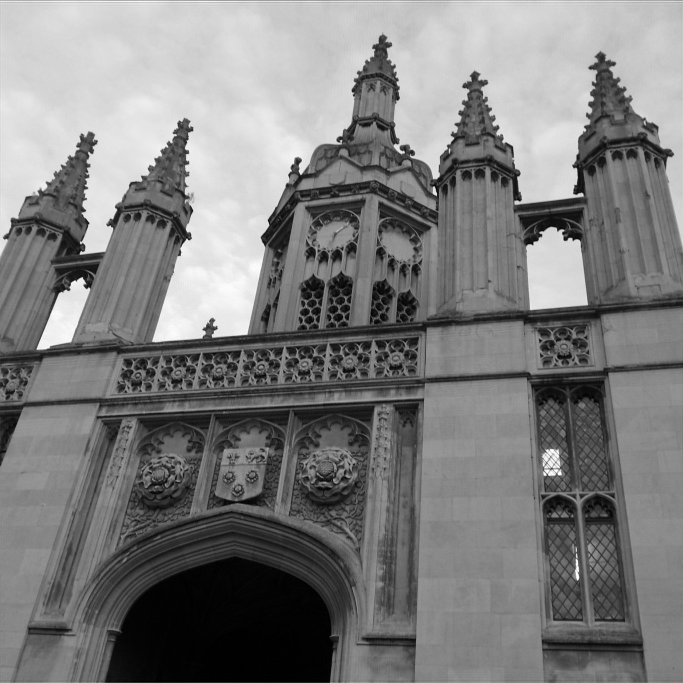}
          &\abrflip{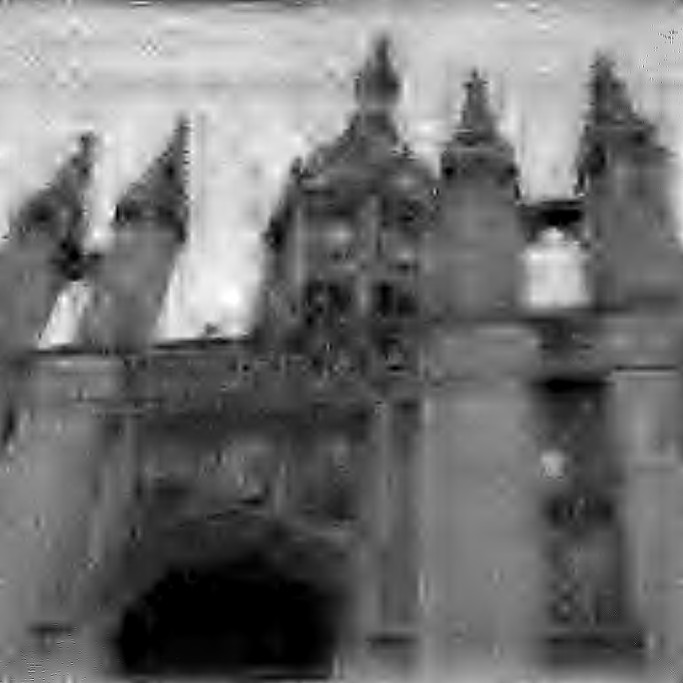}
          &\abrflip{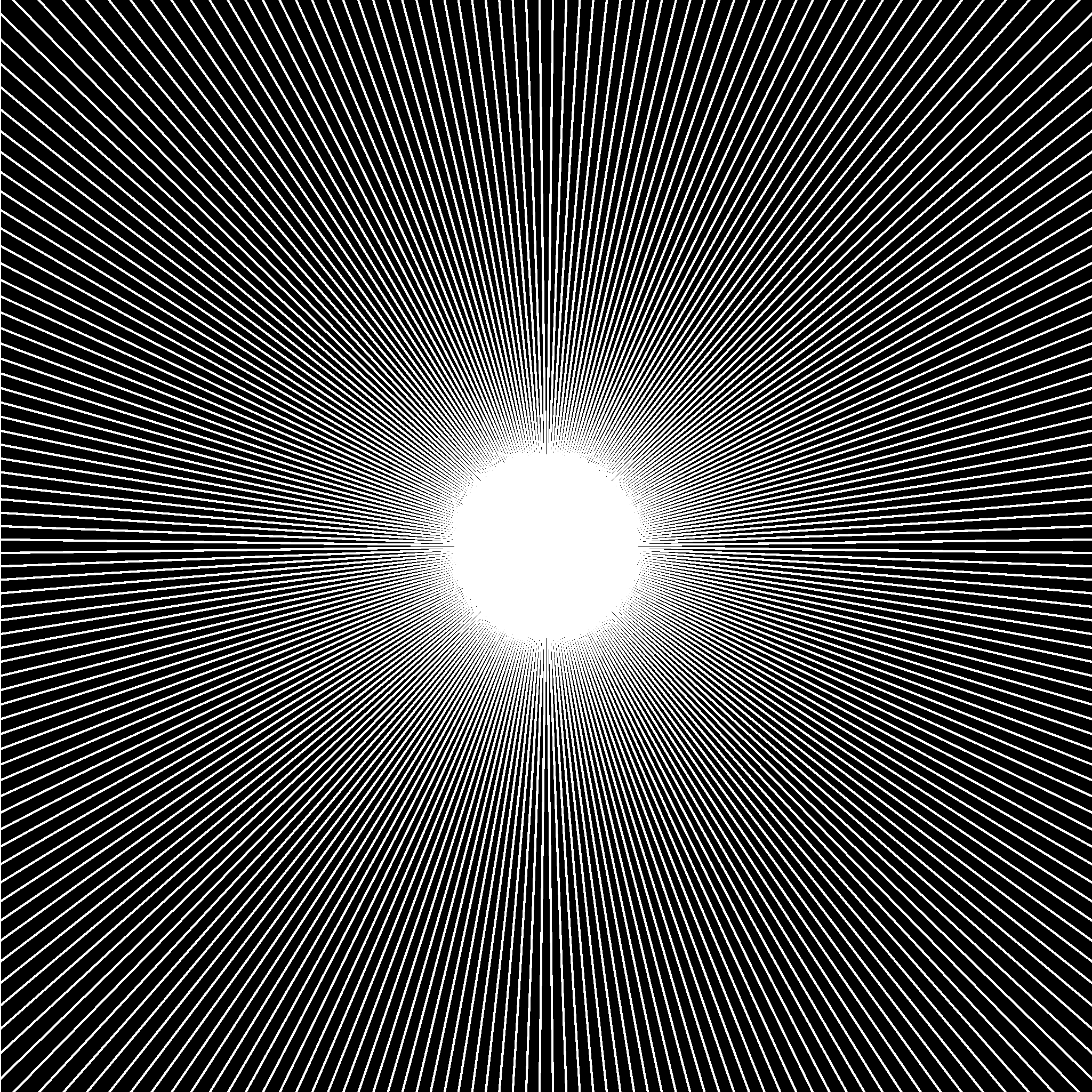}          
\end{tabular}
\caption{CS reconstruction based on $\tilde x^1$ (left column) and the flip test based on $Q_{\text{reverse}} \tilde x^2$ (middle 
column). As in Section \ref{Section:flipTest}, if the RIP holds then the first and middle columns should be similar. The right column shows the subsampling pattern used (with white representing a sample taken at that location). The percentage shown is the fraction of possible measurements that 
were sampled. $\mathrm{DWT}_N$ denotes the transform corresponding to Daubechies wavelets with $N$ vanishing moments.}
\label{fig:flipTestImage}
\end{center}
\end{figure}
\subsection{Why don't we see the RIP?}
\begin{figure}[t]
\centering

\begin{minipage}{.28\textwidth}
\centering
\includegraphics[width=\linewidth]{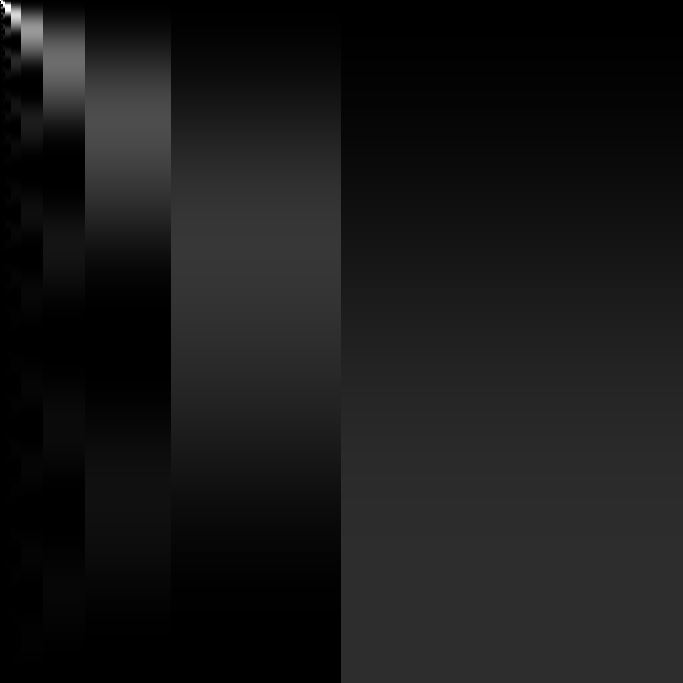}

\end{minipage}\hfill
\begin{minipage}{.28\textwidth}
\centering
\includegraphics[width=\linewidth]{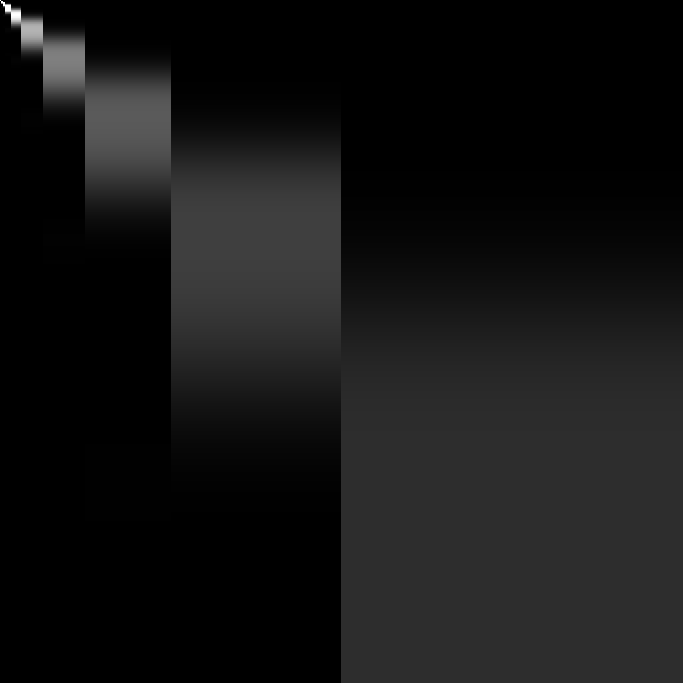}

\end{minipage}\hfill
\begin{minipage}{.28\textwidth}
\centering
\includegraphics[width=\linewidth]{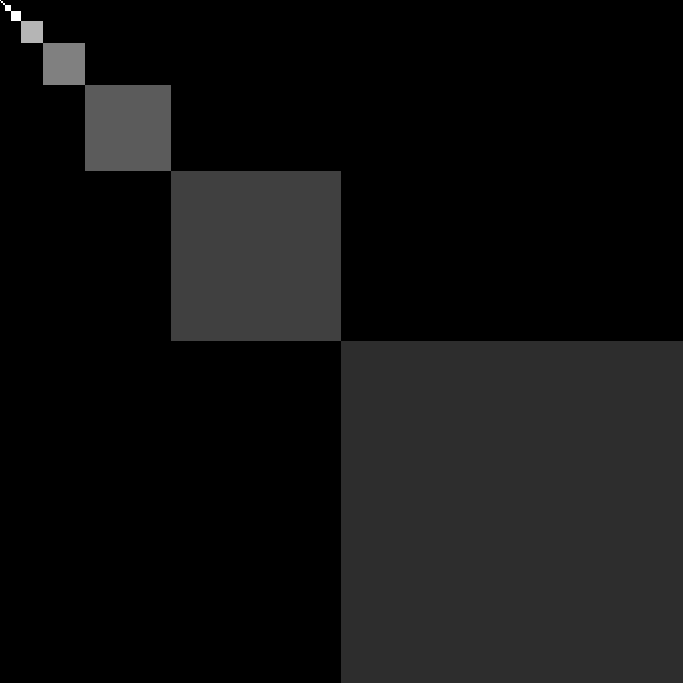}

\end{minipage}\hfill
\caption{The three images display the absolute values of various sensing matrices. A lighter colour represents larger absolute values. (Left) $\mathrm{DFT}\cdot \mathrm{DWT^{-1}_{2}}$, (middle) $\mathrm{DFT}\cdot \mathrm{DWT^{-1}_{10}}$ and (right) $\mathrm{HAD}\cdot \mathrm{DWT^{-1}_{\mathrm{Haar}}}$ where $\mathrm{DFT}$ is the Discrete Fourier Transform, $\mathrm{HAD}$ the Hadamard transform and $\mathrm{DWT^{-1}_{N}}$ the Inverse Wavelet Transform corresponding to Daubechies wavelets with $N$ vanishing moments.}
\label{fig:absMatrices}
\end{figure}
To illustrate the reason for the lack of RIP in the previous examples, let us consider a matrix $M \in \complex^{n \times n}$ such that 
\begin{equation*}
M = \begin{pmatrix} M_1 & 0 &  \dots & 0 & 0 \\ 0 & M_2 & \dots & 0 & 0 \\  \vdots & \vdots & \ddots &\vdots & \vdots \\ 0 & 0 & \dots & M_{l-1} & 0 \\
0 & 0 & \dots & 0 & M_l \end{pmatrix},
\end{equation*}
where the matrices $M_i\in \complex^{a_i}\times \complex^{a_i}$ are isometries and $a_i$ is the length of the $i$th wavelet level. We shall consider the standard compressed sensing problem of recovering $x \in \mathbb{C}^n$, where $n = a_1 + a_2 + \dotsb + a_l$, from knowledge of $P_{\Omega}Mx$ for some sampling set $\Omega$ with $|\Omega| \ll n$. Given the block diagonality of $M$ it is natural to consider a subsampling scheme that fits this structure. In particular, let 
$$
\Omega = \Omega_1 \cup \hdots \cup \Omega_l,  \qquad  \Omega_j \subset \left\{\sum_{k=1}^{j-1} a_k + 1, \hdots,  \sum_{k=1}^{j}a_k\right\}, \quad 2 \leq j \leq l,
$$
and $\Omega_1 \subset \{1,\hdots, a_1\}$. We also, let $
P_{\Omega} = P_{\Omega_1} \oplus \hdots \oplus P_{\Omega_l},
$where the direct sum refers to the obvious separation given by the structure of $M$.
Note that the block diagonal structure of $M$ is a simplified model of the real world, however, as Figure \ref{fig:absMatrices} demonstrates, it is a good approximation when considering the popular matrices
$\mathrm{DFT}\cdot \mathrm{DWT^{-1}}$ and $\mathrm{HAD}\cdot \mathrm{DWT^{-1}}$. In fact, $\mathrm{HAD}\cdot \mathrm{DWT^{-1}_{\mathrm{Haar}}}$ is completely block diagonal.
Suppose that 
$
x = (x^1,x^2,\dotsc,x^l)$ with $x^i = (x^i_1,\hdots, x^i_{a_i}),
$ and that $s_i = |\mathrm{supp}(x^i)|$. If we assume that, for $i=1,2,\dotsc,l$, the matrices $P_{\Omega_i}M_i$ satisfy the RIP with reasonable constants $\delta^i_{s_i}$ then we can easily recover $x$ from $P_{\Omega}Mx$ using $\ell^1$ recovery. 
Set $\tilde{s} = s_1 + s_2 + \dotsb + s_l$. If the global RIP is the explanation for the success of compressed sensing in this case, we would require $\delta_{\tilde{s}}$ to be of reasonable size for $P_{\Omega}M$, so that we can recover  $\tilde x =  Q_{\text{reverse}}x$ from $P_{\Omega}M\tilde x$. Since $\delta_{\tilde{s}} \geq \delta^i_{\tilde{s}}$, then $\delta^i_{\tilde{s}}$ must also be small.

This leads to problems: if $x$ represents the wavelet coefficients of an image then $s_i$ is roughly constant so that $s_i = s$ for some $s \in \mathbb{N}$. A good sampling pattern will have enough measurements so that the finest detail level $P_{\Omega_l}M_l$ can recover $s$-sparse vectors but recovering $s'$-sparse vectors with $s' \gg s$ is impossible. In particular, $ls$ is much larger than $s$ and so $\delta_{\tilde{s}} \geq \delta^l_{\tilde{s}} \approx \delta^l_{ls}$ is too large. Even though the RIP is unable to explain the success of compressed sensing in Figure \ref{fig:flipTestImage}, it is clear that with these examples successful recovery is still possible. We would like to emphasize the point that the RIP is merely a sufficient but not necessary condition for compressed sensing.

\begin{remark}\label{remarkRIPLackExplain}  Although this example is simplified, both regarding the sampling procedure and the number of wavelet levels, similar arguments will explain the lack of RIP for different variable density sampling schemes (see \cite{SMSparsity, BigotBlockCS, chauffert2014variable, VanderEtAlVariable, Siemens, KrahmerWardCSImaging, Lustig3} for more on variable density/multilevel sampling and structured sampling \cite{CalderbankNIPS, CalderbankCommInspCS}). The key is that successful variable density sampling schemes will depend on the distribution of coefficients, just like this simplified example does. Note also that changing from the finite-dimensional model (which is far from optimal for infinite-dimensional problems such as MRI) to the infinite-dimensional model described in \cite{BAACHGSCS, SMSparsity,FundamentalPerformanceLimits, PuyDaviesGribonval2015, GLPU, PruessmannUnserMRIFast} will not fix the problem.
\end{remark}

\subsection{The weighted RIP}

Consideration of a different explanation for the success of 	compressed sensing that includes more structure than just plain sparsity is not a novel idea. Indeed, the \emph{`weighted RIP'} was described in \cite{WeightedL1} as a structured alternative to the RIP. To describe this approach, we shall begin by defining \emph{weighted sparsity}. More specifically, given a collection of weights $\mathbf{\omega}:= (\omega_1,\omega_2 ,\dotsc,\omega_n) \in \mathbb{R}^n$ with $\omega_j \geq 1$ for each $j$, a vector $x \in \complex^n$  is said to be $(\mathbf{\omega},s)$-weighted sparse if the weighted $\ell^0$ norm, $\|x\|_{\mathbf{\omega},0}:= \sum_{j \in \text{supp}(x)} \omega^2_j$, satisfies $\|x\|_{\mathbf{\omega},0} < s$. We can similarly extend the $\ell^1$ norm to a weighted $\ell^1$ norm by defining $\|x\|_{\mathbf{\omega},1}:= \sum_{j=1}^n \omega_j x_j$ and then examine \emph{weighted $\ell^1$ minimisation} in the same way that we can discuss $\ell^1$ minimisation. A preliminary idea to deal with the difficulties raised in section \ref{Section:flipTest} is to argue that instead of expecting to recover all $s$-sparse vectors, we should expect to recover all \weightedsparse vectors. This is further motivated by the success of such an approach to the recovery of smooth functions from undersampled measurements \cite{WeightedL1}.

\subsubsection{The weighted RIP is insufficient with wavelets}
Unfortunately, there are issues with this approach when applied to problems involving a level based construction basis such as wavelets like in section \ref{Section:flipTest}. These are more thoroughly documented in \cite{GeneralisedFlipTest}, but we shall provide a brief outline here. Just as the flip test demonstrates that in many examples relevant to practical applications the class of $s$-sparse problems is too big and contains vectors that cannot be recovered by $\ell^1$ minimisation, we have the same phenomenon for weighted sparsity. Typically, for problems involving a level based reconstruction basis, and for any choice of weights $\omega$, the class of \weightedsparse vectors is too large and contains vectors that cannot be recovered by either weighted-$\ell^1$ or $\ell^1$ minimisation. In our specific setting above, this means that we have a `natural' image with wavelet coefficients $w$ that is recovered exactly and a vector $w'$ with $\|w'\|_{(\mathbf{\omega},0)} \leq \|w\|_{(\mathbf{\omega},0)}$ which is not recovered. Therefore, either $\|w\|_{(\mathbf{\omega},0)} > s$ and $w$ is not \weightedsparse (so a theory based on weighted sparsity does not explain why $w$ is recovered) or $w'$ is \weightedsparse (but not recovered, so that the class of \weightedsparse vectors is too large).

We can show these results by expanding the `flip test' from section \ref{Section:flipTest}. The result is the \emph{generalised flip test}:
\begin{enumerate}
\item Let $U \in \complex^{m \times n}$ be the combined measurement operator and sparsifying transformation (e.g. $U = P_{\Omega}\text{DFT}\cdot\text{DWT}^{-1}$) and let $w$ be a `typical' vector that can be recovered by $U$ (using either $\ell^1$ or weighted-$\ell^1$ minimisation).
\item Threshold $w$ so that it is perfectly sparse and perfectly recovered and then set $w^1$ to be the vector with $w^1 = 1$ if $w_i \neq 0$ and $w_i = 0$ otherwise. Since there is no change in the location of the non-zero coefficients, $w^1$ can be perfectly recovered.

\item For a given collection of weights $\omega$, set $s$ to be the minimal value so that $w^1$ is \weightedsparse. If weighted sparsity is the right model for compressed sensing then all \weightedsparse vectors should be recovered exactly.
\item We now choose a vector $w^2$ with $\|w^2\|_{\mathbf{\omega},0} \leq \|w^1\|_{\mathbf{\omega},0}$, with each non-zero entry of $w^2$ set to $1$. We do this in such a way that $w^2$ has a larger number of non-zero entries contained in some level $l_0$ than $w^1$. Since $\|w^2\|_{\mathbf{\omega},0} \leq \|w^1\|_{\mathbf{\omega},0}$, $w^2$ is \weightedsparse and should be recovered exactly. If this does not happen, then weighted sparsity is not a complete description of the success of compressed sensing when applied to this specific $U$.
\end{enumerate}

As we can see in Figure \ref{fig:WeightedRIPL1}, there are many examples where weighted sparsity is insufficient to explain the success of compressed sensing. We shall provide additional insight as to why weighted sparsity is insufficient in Section \ref{Section:WeightedSparsityAndSMSparsity}.

\begin{remark}\label{remark:WeightedL1Uses}
It must be emphasised that weighted sparsity and the weighted RIP were developed in \cite{WeightedL1} for the purpose of recovering smooth functions with polynomials. Thus, one should not expect the weighted RIP to hold for wavelets.  Conversely, the RIP in levels may not work for polynomials, as unlike wavelets there is no level structure present. This fact demonstrates the finesses of compressed sensing theory and that we are in need for a collection of much more specific theorems using different sparsity models that depend on the problem.
\end{remark}

\begin{figure}
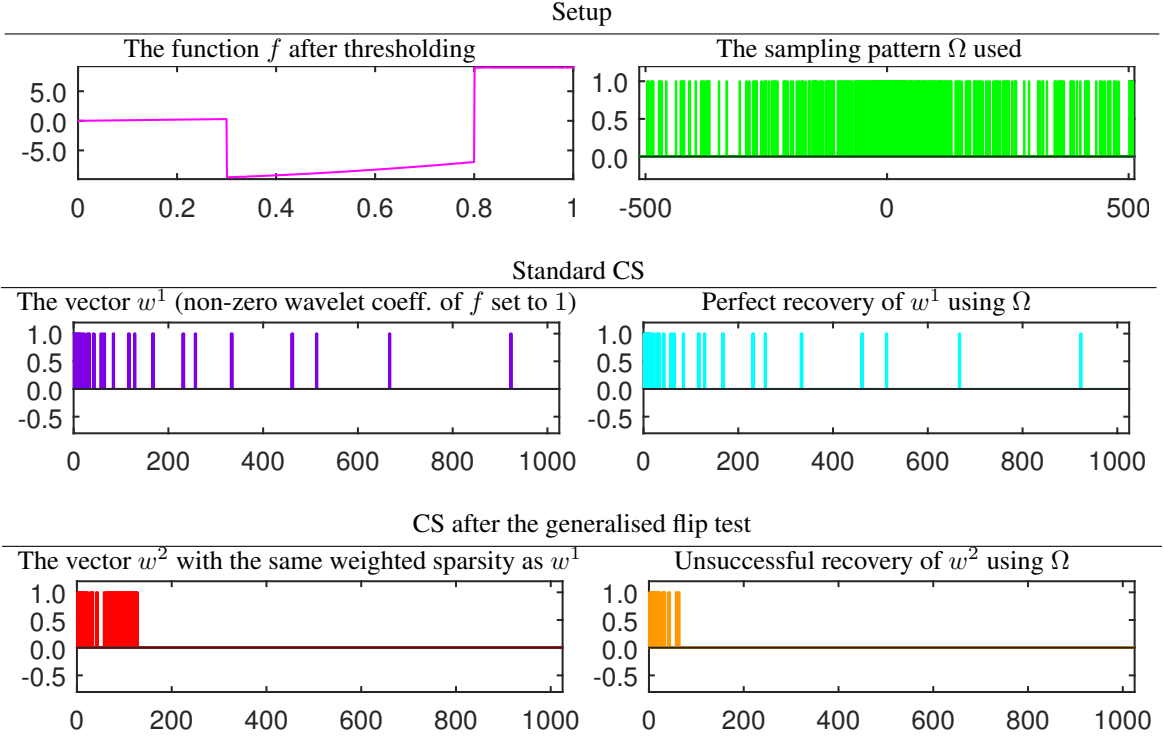

\begin{center}
\newcommand{\abrcap}[2]{\raisebox{80pt}{\begin{minipage}[t]{65pt}{\footnotesize
 #1x#1\\[15pt]Error:\\[2pt]}#2\%\end{minipage}}}%
\newcommand{\abrdesc}[5]{\raisebox{100pt}{\begin{minipage}[t]{60pt}{\footnotesize
  \res{#1}\\#2\%\\[7pt]#3$\cdot$#4\\[7pt]{#5}}\end{minipage}}}%
\newcommand{\abrmain}[1]{\includegraphics[width=0.45\textwidth]{#1}}%
\vspace{-8pt}
\centering 
Setup

\vspace{2pt}
\begin{tabular}[h]{c@{\hspace{4pt}}c@{\hspace{4pt}}c@{\hspace{0pt}}}
\hline
	The function $f$ after thresholding & The sampling pattern $\Omega$ used\\
	\abrmain{/WFlip/ExperN3224_SparseImage.pdf}
	&   \abrmain{/WFlip/ExperN3224_SP.pdf}
	\\
	
\end{tabular}

\vspace{10pt}
Standard CS
\vspace{2pt}

\begin{tabular}[h]{c@{\hspace{4pt}}c@{\hspace{4pt}}c@{\hspace{0pt}}}
\hline
	
		The vector $w^1$ (non-zero wavelet coeff. of $f$ set to $1$) & Perfect recovery of $w^1$ using $\Omega$\\
		\abrmain{/WFlip/ExperN3224_OrigWC.pdf}
	&\abrmain{/WFlip/ExperN3224_OrigWCRec.pdf}
	\\ 
\end{tabular}
\vspace{10pt}

CS after the generalised flip test

\vspace{2pt}
\begin{tabular}[h]{c@{\hspace{4pt}}c@{\hspace{4pt}}c@{\hspace{0pt}}}
\hline
	The vector $w^2$ with the same weighted sparsity as $w^1$ & Unsuccessful recovery of $w^2$ using $\Omega$\\	
	\abrmain{/WFlip/ExperN3224_MovedWC.pdf}
	&\abrmain{/WFlip/ExperN3224_RecoveredWC.pdf}
	\\ 
\end{tabular}
\begin{center}
\end{center}
\end{center}
\caption{The figure displays the generalised fliptest for weighted sparsity with the function $f(x) = \sin(x) \mathds{1}_{[0,0.3]} -10\cos(x) \mathds{1}_{(0.3,0.8]} + 9\mathds{1}_{(0.8,1]}$. Recovery was done using a subsampled $1$D fourier to wavelet matrix, with Daubechies $3$ wavelets and $\ell^1$ minimisation. The weights on the coefficients in level $i$ were given by $2^i$. Similar results follow for other weights and also for recovery with weighted $\ell^1$ minimisation.} 
\label{fig:WeightedRIPL1}
\end{figure}

\section{An extended theory for compressed sensing}
The current mathematical theory for compressive sensing revolves around a few key ideas. These are the concepts of sparsity, incoherence, uniform subsampling and the RIP. In \cite{SMSparsity} and \cite{AsymptoticStructure}, it was shown that these concepts are absent for a large class of compressed sensing problems. To solve this problem, the extended concepts of \emph{asymptotic sparsity}, \emph{asymptotic incoherence} and \emph{multi-level sampling} were introduced. We now introduce the fourth extended concept in the new theory of compressive sensing: the \emph{RIP in levels}. This generalisation of the RIP replaces the idea of \emph{uniform recovery of $s$-sparse signals} with that of \emph{uniform recovery of $(s,M)$-sparse signals}. We shall detail this new idea in this section.

\subsection{A level based alternative: the RIP in levels}\label{Section:RIPL}
The examples given in Figure \ref{fig:flipTestImage} all involve reconstructing in a basis that is divided into various levels. It is this level based structure that allows us to introduce a new concept which we term the `RIP in levels'. We shall demonstrate such a concept based on our observations with wavelets, which we give a brief description of in the following section. Despite our focus on wavelets in the next few pages, it should be noted that our work applies equally to all level based reconstruction bases. 
\subsubsection{Wavelets}
 A multiresolution analysis (as defined in \cite{Ingrid, Ingrid2,  MaaB}) for $L^2(X)$ (where $X$
 is an interval or a square) is formed by constructing increasing scaling spaces $(V_j)_{j=0}^{\infty}$ and wavelet spaces $(W_j)_{j=0}^{\infty}$ with $V_j,W_j \subset L^2(X)$  so that
\begin{enumerate}
\item If $f(\cdot) \in V_j$ then $f(2\cdot) \in V_{j+1}$, and vice-versa.   \label{MRADilationInvariance}
\item $\overline{\bigcup \limits_{j = 0}^{\infty} V_j} = L^2(X)$ and $\bigcap \limits_{j=0}^{\infty} V_j = \left\lbrace 0 \right\rbrace$.
\item $W_j$ is the orthogonal complement of $V_j$ in $V_{j+1}$.
\end{enumerate}
For natural images $f$, the largest coefficients in the wavelet expansion of $f$ appear in the levels corresponding to smaller $j$.
 Closer examination of the relative sparsity in each level also reveals a pattern: let $w$ be the collection of wavelet coefficients of $f$ and for a given level $k$ let $S^k$ be the indices of all wavelet coefficients of $f$ in the $k$th level. Additionally, let $\mathcal{M}_n$ be the largest (in absolute value) $n$ wavelet coefficients of $f$. Given $\epsilon \in [0,1]$, we define the functions $s(\epsilon)$ and $s_k(\epsilon)$ (as in \cite{SMSparsity}) by
\begin{align*}
s(\epsilon) &:= \min \left \lbrace n : \left \| \sum_{i \in \mathcal{M}_n} w_i \right\| \geq \epsilon \,\left \|\sum_{k=0}^{\infty} \sum_{i \in S^k} w_i \right\|\right\rbrace, \quad s_k(\epsilon) := |\mathcal{M}_{s(\epsilon)} \cap S^k|.
\end{align*}
\begin{equation*} 
\end{equation*}
More succinctly, $s_k(\epsilon)$ represents the relative sparsity of the wavelet coefficients of $f$ at the $k$th scale. If an image is very well represented by wavelets, we would like $s_k(\epsilon)$ to be as small as possible for $\epsilon$ close to 1. However, one can make the following observation: if we define the numbers $M_k$ so that $M_k - M_{k-1} = |S^k|$ (the reason for this choice of notation will become clear in Section \ref{section:FlipTestInLevels}) then the ratios $\frac{s_k(\epsilon)}{M_k - M_{k-1}}$ decay very rapidly for a fixed $\epsilon$. Numerical examples showing this phenomenon with Haar Wavelets are displayed in Figure \ref{fig:skeps}. Summarising, we observe that the sparsity of typical images has a structure which the traditional RIP ignores. \vspace{10pt}

\begin{figure}
\begin{center}
\newcommand{\abrcap}[2]{\raisebox{80pt}{\begin{minipage}[t]{65pt}{\footnotesize
 #1x#1\\[15pt]Error:\\[2pt]}#2\%\end{minipage}}}%
\newcommand{\abrdesc}[5]{\raisebox{100pt}{\begin{minipage}[t]{60pt}{\footnotesize
  \res{#1}\\#2\%\\[7pt]#3$\cdot$#4\\[7pt]{#5}}\end{minipage}}}%
\newcommand{\abrflip}[1]{\includegraphics[height=0.27\textwidth]{#1}}%
\newcommand{\abrmain}[1]{\includegraphics[width=0.27\textwidth]{#1}}%
\vspace{-10pt}
\begin{tabular}[h]{c@{\hspace{4pt}}c@{\hspace{4pt}}c@{\hspace{0pt}}}
	
	\abrmain{image_College1GrayCS}
	&\abrflip{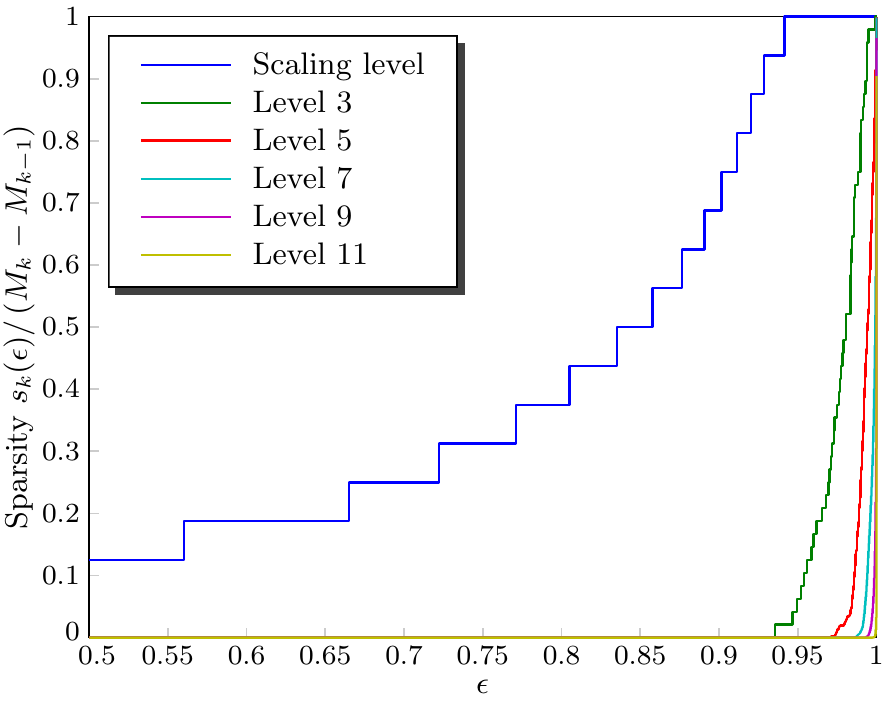}
	&\abrflip{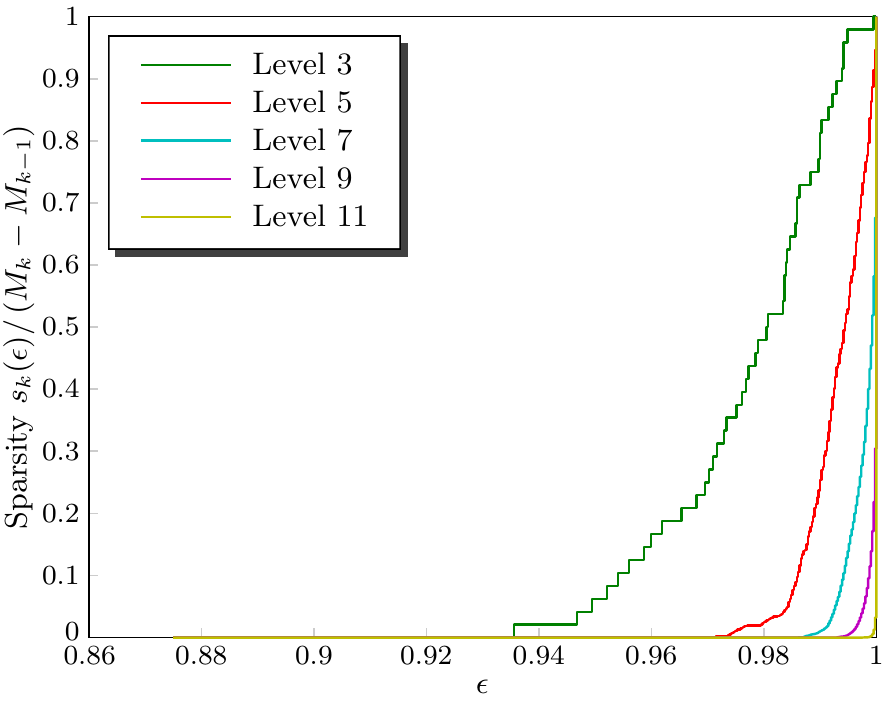}
	\\	
	\abrmain{image_RocksGrayCSHad}
	&\abrflip{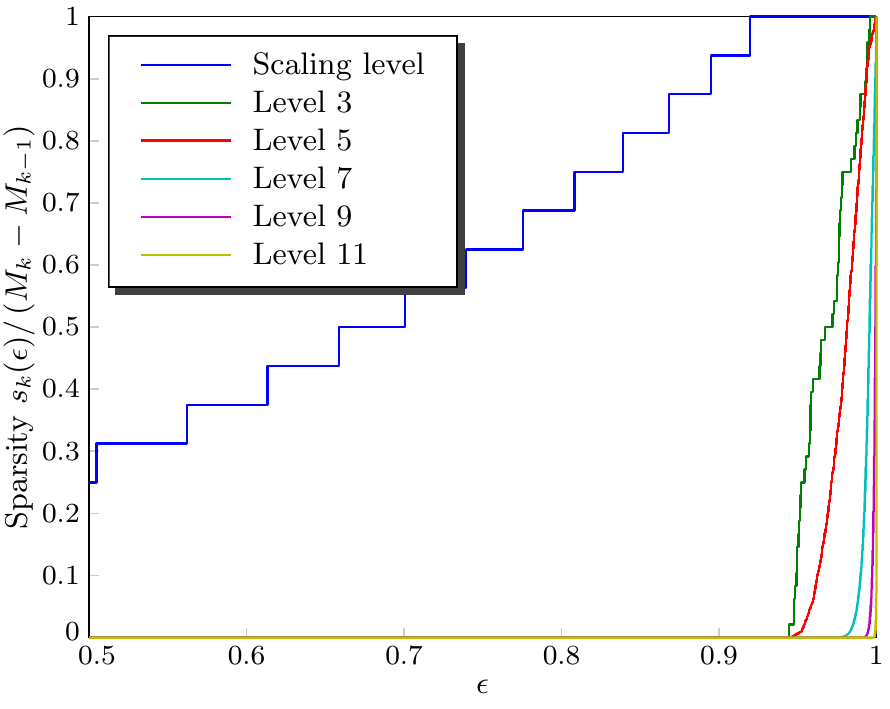}
	&\abrflip{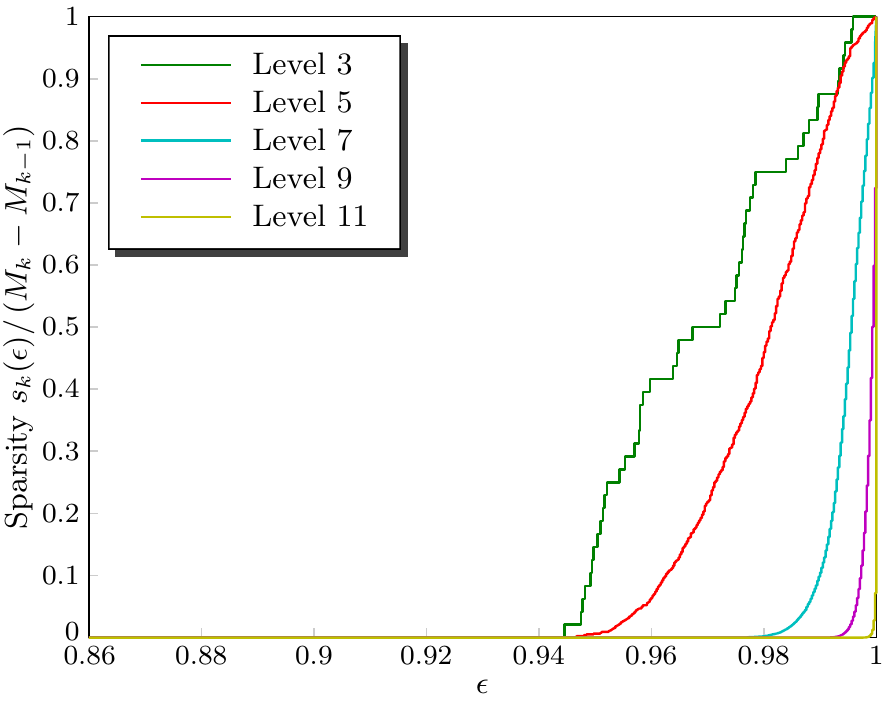}
	\\ 
\end{tabular}
\caption{The relative sparsity of Haar wavelet coefficients of two image. The leftmost column displays the image in question. The middle and final columns display the values of $s_k(\epsilon)$ for $\epsilon \in [0.5,1]$ and $\epsilon \in [0.86,1]$ respectively, where $k$ represents a wavelet level. Of particular importance is the rapid decay of $s_k(\epsilon)$ as $k$ grows larger.}
\label{fig:skeps}
\end{center}
\end{figure}

\subsubsection{$\smsparse$-sparsity and the flip test in levels}\label{section:FlipTestInLevels}
Theorem \ref{Theorem:RIPtypeTheorem} and any similar theorems all suggest that we are able to recover all $s$-sparse vectors exactly, independent of which levels the $s$-sparse vectors are supported on.  Instead of such a stringent requirement, we can take advantage of the structure of our problem, a concept that is already popular from the recovery point of view \cite{BaranuikModelCS, Volkan, HeCarinStructCS, HeCarinTreeCS}.  We have observed that, for wavelets,  as $k \to \infty$ $s_k(\epsilon)/(M_k - M_{k-1})$ decays rapidly (see Figure \ref{fig:skeps}). To further understand this phenomenon, in \cite{SMSparsity} the concept of $\smsparse$-sparsity was introduced.

\begin{definition}[$\smsparse$-sparsity]\label{def:SMSPARSEDefinition}
Let $\mathbf{M}=(M_0,M_1, \dotsc, M_{l}) \in \mathbb{N}^{l+1}$ with $1 \leq M_1 < M_2 < \dotsb <M_{l}$ and $M_0 = 0$, where the natural number $l$ is called the \emph{number of levels}. Additionally, let $
\mathbf{s}=(s_1,s_2, \dotsc, s_l) \in \mathbb{N}^l$ with $s_i \leq M_i - M_{i-1}.$ We call $\smsparse$ a \emph{sparsity pattern}.

A set $\Lambda$ of integers is said to be $\smsparse$-sparse if 
$\Lambda \subset \left\lbrace M_0+1,M_0+2, \dotsc, M_l \right\rbrace 
$ and for each $i \in \left\lbrace 1,2 ,\dotsc, l\right\rbrace$, we have
$
|\Lambda \cap \left\lbrace M_{i-1}+1,M_{i-1}+2, \dotsc, M_{i} \right\rbrace|  \leq s_i.
$
A vector is said to be $\smsparse$-sparse if its support is an $\smsparse$-sparse set.
The collection of $\smsparse$-sparse vectors is denoted by $\Sigma_{\mathbf{s,M}}$.
We can also define $\sigma_{\mathbf{s,M}}(x)_1$ as a natural extension of $\sigma_{s}(x)_1$. Namely,
$
\sigma_{\mathbf{s,M}}(x)_1:= \min\limits_{\widehat{x} \in \Sigma_{\mathbf{s,M}}} \|x - \widehat{x}\|_1. $
\end{definition}

\begin{remark}\label{rem:SMSPARSETechnicality}
If $\smsparse$ is a sparsity pattern, we will sometimes refer to $(a\mathbf{s,M})$-sparse sets for some natural number $a$ even though $as_i$ may be larger than  $M_i - M_{i-1}$. To make sense of such a statement, we define (in this context) 
$$
a\mathbf{s} := \left(\min(as_1,M_1 - M_0),\min(as_2,M_2 -M_1), \dotsc ,\min(as_l,M_{l}-M_{l-1})\right).
$$
\end{remark}
Let us now look at a specific case where $\smsparse$ represent wavelet levels (again, we emphasise that wavelets are simply one example of a level based system and that our work is more general). Roughly speaking, we can choose $\mathbf{s}$ and $\mathbf{M}$ such that $x$ is $\smsparse$-sparse if it has fewer non-zero coefficients in the finer wavelet levels. As in Theorem \ref{Theorem:RIPtypeTheorem} (where we take $\epsilon = 0$ corresponding to an absence of noise), we shall examine solutions $\tilde{x}$ to the problem 
\begin{equation}\label{eq:l1ProblemClose}
\min \|\widehat{x}\|_1 \text{ such that } Ux = U\widehat{x}.
\end{equation} 
Instead of asking for $
\|x-\tilde{x}\|_1 \leq C \sigma_{s}(x)_1
$ and $
\|x-\tilde{x}\|_2 \leq C \sigma_{s}(x)_1/\sqrt{s},$
we might instead look for a condition on $U$ that allows us to conclude that
\begin{align}
\|x-\tilde{x}\|_1 &\leq C \sigma_{\mathbf{s,M}}(x)_1 \,\,\,\text{and} \label{equation:conclusionOfRIPInLevelsTypeTheorem}\\
\|x-\tilde{x}\|_2 &\leq  \frac{C\sigma_{\mathbf{s,M}}(x)_1}{\sqrt{s_1 + s_2 + \dotsb + s_l}}, \label{equation:conclusionOfRIPInLevelsTypeTheoremL2}
\end{align}
for some constant $C$, independent of $x$ and $s,\eta_{\mathbf{s,M}}$. In Section \ref{Section:flipTest}, we saw that there was a simple test that was able to tell us if a matrix does not exhibit the RIP. However, the argument in Section \ref{Section:flipTest} does not hold if we only insist on results of a form similar to equations \eqref{equation:conclusionOfRIPInLevelsTypeTheorem} and \eqref{equation:conclusionOfRIPInLevelsTypeTheoremL2}. Instead, we can describe a `flip test in levels' in the following way:
\begin{enumerate}
\item Suppose that a matrix $U \in \complex^{m \times n}$ satisfies \eqref{equation:conclusionOfRIPInLevelsTypeTheorem} and \eqref{equation:conclusionOfRIPInLevelsTypeTheoremL2} when solving an $\ell^1$ minimisation problem of the same form as \eqref{eq:l1ProblemClose}. Take a sample vector $x^1$ and compute $y^1:=Ux^1$. Set $\widetilde{x}^1$ to be the solution to \eqref{eq:l1ProblemClose} where we seek to reconstruct $x^1$ from $y^1$. 
%\item Because the property of being $s$-sparse is independent of ordering, if \eqref{ripDefinition} holds for $U$ then it also holds for the matrix $UP$ where $P$ is a permutation matrix. 
\item For an operator $Q$ which permutes entries of vectors in $\complex^n$ such that $Q(\Sigma_{\mathbf{s,M}}) = \Sigma_{\mathbf{s,M}}$, set $x^2:=Qx^1$, and compute $y^2 := Ux^2$. Again, set $\widetilde{x}^2$ to be the solution to \eqref{eq:l1ProblemClose} where we seek to reconstruct $x^2$ from $y^2$.
\item $Q(\Sigma_{\mathbf{s,M}}) = \Sigma_{\mathbf{s,M}}$ implies that $C \sigma_{\mathbf{s,M}}(x^1)_1  = C \sigma_{\mathbf{s,M}}(x^2)_1$, because
\begin{align*}
\sigma_{\mathbf{s,M}}(x^2)_1&= \min\limits_{\widehat{x}^2 \in \Sigma_{\mathbf{s,M}}} \|x^2 - \widehat{x}^2\|_1 = \min\limits_{\widehat{x}^2 \in Q(\Sigma_{\mathbf{s,M}})} \|x^2 - \widehat{x}^2\|_1\\
&= \min\limits_{\widehat{x}^1 \in \Sigma_{\mathbf{s,M}}} \|x^2 - Q\widehat{x}^1\|_1 = \min\limits_{\widehat{x}^1 \in \Sigma_{\mathbf{s,M}}} \|Qx^1 - Q\widehat{x}^1\|_1\\
&= \min\limits_{\widehat{x}^1 \in \Sigma_{\mathbf{s,M}}} \|x^1 - \widehat{x}^1\|_1 = \sigma_{\mathbf{s,M}}(x^1)_1
\end{align*}
where we have used the fact that permutations are isometries in the transition to the final line.
\item From \eqref{equation:conclusionOfRIPInLevelsTypeTheorem}, we have $\|x^2-\widetilde{x}^2\|_1 \leq C \sigma_{\mathbf{s,M}}(x^1)_1$. Again, since permutations are isometries, we see that $\|x^1-Q^{-1}\widetilde{x}^2\|_1 \leq C \sigma_{\mathbf{s,M}}(x^1)_1.$
\item If $x^1$ is well approximated by $\widetilde{x}^1$ and equations \eqref{equation:conclusionOfRIPInLevelsTypeTheorem} and \eqref{equation:conclusionOfRIPInLevelsTypeTheoremL2} are in some sense optimal, we expect 
$\|x^1-\widetilde{x}^1\|_1 \approx C \sigma_{\mathbf{s,M}}(x^1)_1 $. Thus,  $x^1$ should also be well approximated by $Q^{-1}\widetilde{x}^2$ since $\|x^1 - Q^{-1}\widetilde{x}^2\|_1 \leq C \sigma_{\mathbf{s,M}}(x^2)_1 =  C \sigma_{\mathbf{s,M}}(x^1)_1$.
\item If this fails, then equations \eqref{equation:conclusionOfRIPInLevelsTypeTheorem} and \eqref{equation:conclusionOfRIPInLevelsTypeTheoremL2} cannot be the reason for the success of compressed sensing with the matrix $U$.
\end{enumerate}
 \begin{table}
 \caption{Flip test in levels with randomly generated permutations\label{table:flipInLevelsMassPermutations}}
 \begin{center}
 \begin{tabular}[H]{ |c|c|c|c|c| }
 \hline
 Image & Subsampling percentage & Max & Min & Standard deviation \\ \cline{1-5}
  College 1 & $12.48\%$ & $4.0560\%$ & $3.8955\%$ & $0.0642\%$\\ \cline{1-5}
  College 2 & $97.17\%$ & $0.4991\%$ & $0.4983\%$ & $0.0001\%$\\ \cline{1-5}
  Rocks & $21.32\%$   & $5.6943\%$ & $5.6049\%$ & $0.0115\%$ \\ \cline{1-5}
  College 3 & $27.51\%$ & $2.0074\%$ & $1.9860\%$ & $0.0032\%$ \\ \cline{1-5}
 \hline
 \end{tabular}
 \vspace{10pt}
 \caption*{A table displaying relative error percentages $\|x-\widetilde{x}\|_2/\|x\|_2$ for various images $x$ as in Figure \ref{fig:flipInLevelsTestImage} with recovered image $\widetilde{x}$. Each image was processed with a fixed subsampling pattern and $1000$ randomly generated permutations as described in Section \ref{section:FlipTestInLevels}. The columns labelled `Max', `Min' and `Standard deviation' list the maximum, minimum and standard deviation of the relative errors taken over all tested permutations. }
 \end{center}
 \end{table}
The requirement that $Q(\Sigma_{\mathbf{s,M}}) = \Sigma_{\mathbf{s,M}}$ now requires us to consider different permutations than a simple reverse permutation as in Section \ref{Section:flipTest}. A natural adaptation of $Q_{\text{reverse}}$ to this new `flip test in levels' is a permutation that just reverses coefficients within each wavelet level. 
Figure \ref{fig:flipInLevelsTestImage} displays what happens when we attempt to do the flip test with this permutation. In this case, we see that the performance of CS reconstruction under flipping and the performance of standard CS reconstruction are very similar. This suggests that uniform recovery within the class of $\smsparse$-sparse vectors (as in \ref{equation:conclusionOfRIPInLevelsTypeTheorem} and \ref{equation:conclusionOfRIPInLevelsTypeTheoremL2}) is possible with a variety of practical compressive sensing matrices. Indeed, in Table \ref{table:flipInLevelsMassPermutations} we also consider a collection of randomly generated $Q$ with $Q(\Sigma_{\mathbf{s,M}}) = \Sigma_{\mathbf{s,M}}$. We see that once again, performance with permutations within the levels is similar to standard CS performance. 

\begin{figure}
\begin{center}
\newcommand{\abrcap}[2]{\raisebox{80pt}{\begin{minipage}[t]{65pt}{\footnotesize
 #1x#1\\[15pt]Error:\\[2pt]}#2\%\end{minipage}}}%
\newcommand{\abrdesc}[6]{\raisebox{110pt}{\begin{minipage}[t]{60pt}{\footnotesize
  #6\\\res{#1}\\#2\%\\[4pt]#3$\cdot$#4\\[4pt]{#5}}\end{minipage}}}%
\newcommand{\abrflip}[1]{\includegraphics[width=0.28\textwidth]{#1}}%
\vspace{-10pt}
\begin{tabular}{@{\hspace{0pt}}c@{\hspace{0pt}}c@{\hspace{3pt}}c@{\hspace{3pt}}c@{\hspace{0pt}}}
& CS reconstruction & CS w/ flip in levels & Subsampling pattern 
used\\[5pt]
  \abrdesc{2048}{12}{$\mathrm{DFT}$}{$\mathrm{DWT}^{-1}_3$}{MRI,\\ Spectroscopy, Radio-\\interferometry}{College 1}
   &\abrflip{image_College1GrayCS}
  &\abrflip{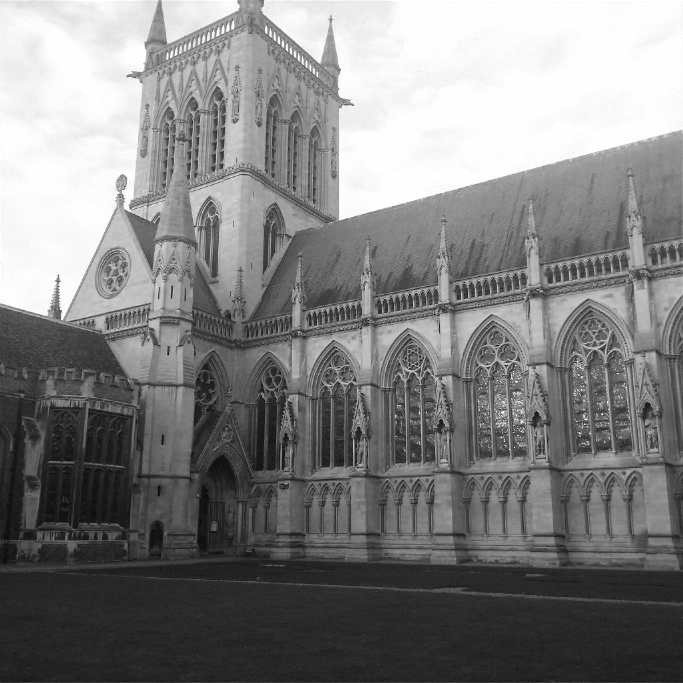}
  &\abrflip{image_College1GraySampling}
 \\
   \abrdesc{2048}{97}{$\mathrm{DFT}$}{$\mathrm{DWT}^{-1}_4$}{MRI,\\ Spectroscopy, Radio-\\interferometry}{College 2}
    &\abrflip{image_College3GrayCS}
      &\abrflip{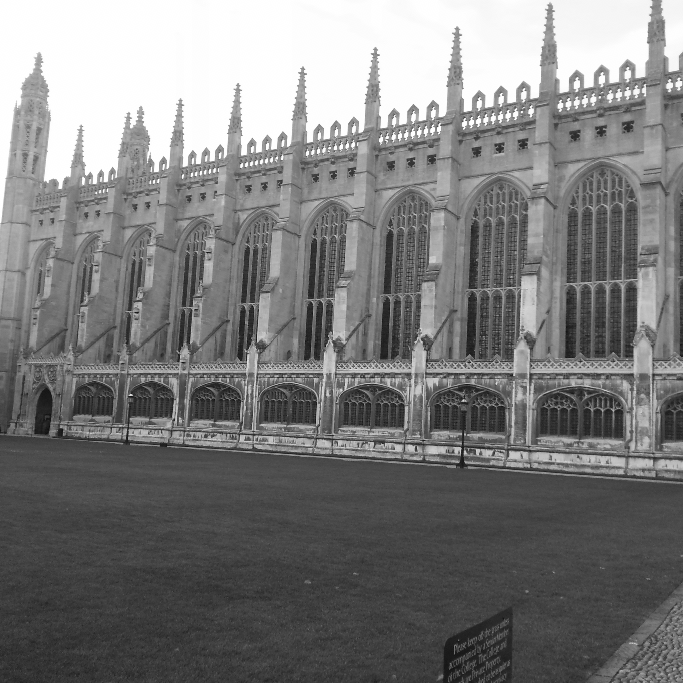}
      &\abrflip{image_College3GraySampling}
  \\
     \abrdesc{2048}{12}{$\mathrm{HAD}$}{$\mathrm{DWT}^{-1}_2$}{Comp. imag.,\\ Hadamard spectroscopy, Fluorescence microscopy}{Rocks}
      &\abrflip{image_RocksGrayCSHad}
        &\abrflip{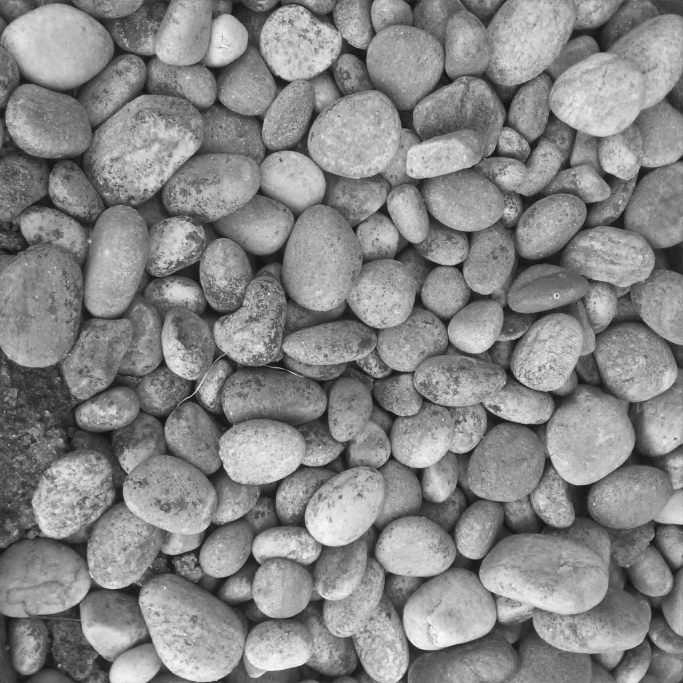}
        &\abrflip{HadamardSampling}
 \\
      \abrdesc{2048}{27.5}{$\mathrm{DFT}$}{$\mathrm{DWT}^{-1}_3$}{Electron \\microscopy,\\ Computerised\\ tomography}{College 3}
       &\abrflip{image_College4RRCS}
         &\abrflip{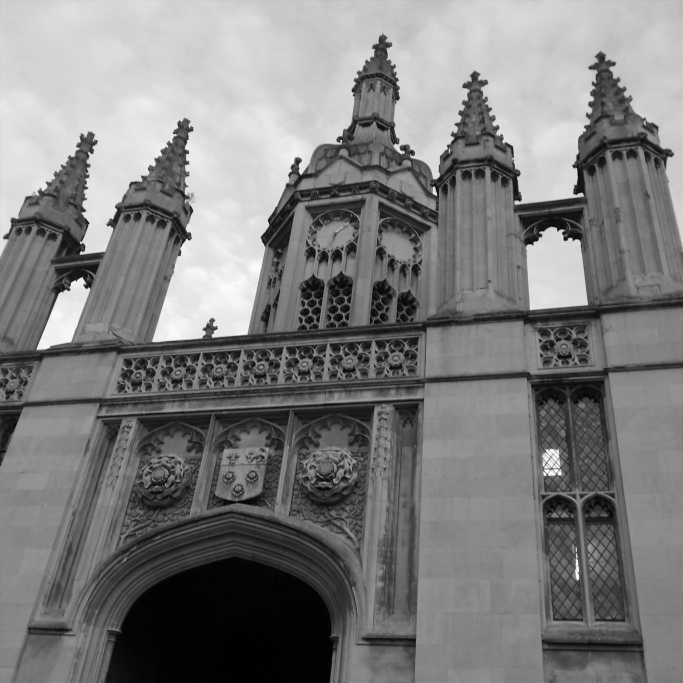}
         &\abrflip{image_College4SP}
\end{tabular}
\caption{Standard reconstruction (left column) and the result of the flip test where the wavelet coefficients are flipped within each level (middle 
column). The right column shows the subsampling pattern used. The percentage shown 
is the fraction of Fourier or Hadamard coefficients that 
were sampled. $\mathrm{DWT}^{-1}_N$ denotes the Inverse Wavelet Transform corresponding to Daubechies wavelets with $N$ vanishing moments.}
\label{fig:flipInLevelsTestImage}
\end{center}
\end{figure}
\subsubsection{Relating $(\mathbf{s,M})$-sparsity and weighted sparsity}\label{Section:WeightedSparsityAndSMSparsity}
The `flip test in levels' suggests that for many compressed sensing problems, there are $\mathbf{s}$ and $\mathbf{M}$ such that all $\smsparse$-sparse vectors are recovered equally well by $\ell^1$ minimisation. With this in mind, we are now in a position to provide additional details on why the same is not the case for weighted sparsity. Indeed, one can easily state and prove the following theorem (see \cite{GeneralisedFlipTest}
for details):
		\begin{theorem}\label{WRIProblems}
	Let $\mathbf{s,M}$ have $l$ levels and fix $r < l$. Suppose that the collection of $\smsparse$-sparse vectors are all $(\mathbf{\omega},X)$-weighted sparse for some $X$. Then there is an $l_0$ with $r < l_0 < l$ such that the collection of $(\mathbf{\tilde{\mathbf{s}},M})$-sparse are also $(\mathbf{\omega},X)$-weighted sparse, where $$\tilde{\mathbf{s}} = (\underbrace{s_1,s_2,\dotsc, s_r}_{r},\underbrace{0,0,\dotsc,0}_{l_0-1-r}, (l-r)s_{l_0},0,\dotsc,0).$$
	\end{theorem}
The use of this theorem becomes apparent if we consider Figure \ref{fig:WeightedRIPL1}. Recall that the Fourier to Wavelet matrix in Figure \ref{fig:WeightedRIPL1} is well approximated by block diagonal matrices (c.f. Figure \ref{fig:absMatrices}). This block diagonality structure means that we can design our sampling pattern so that information corresponding to coarser wavelet levels is more readily captured than the information corresponding to the finer wavelet levels. Typically the first $r$ levels will be fully sampled, but after that subsampling occurs and this is where we run into difficulties with weighted sparsity. If we suppose that recovering all vectors with $s_k$ non zero coefficients in the indices corresponding to the $k$th wavelet level takes $\Omega_k$ measurements in that level, then recovering all weighted sparse vectors requires $(l-r)\Omega_{k}$ measurements for some $k$. Unfortunately, this leads to weighted sparsity overestimating the number of measurements required to recover all vectors of interest.

\subsubsection{The RIP in levels}
Given the success of the `flip test in levels', let us now try to find a sufficient condition on a matrix $U \in \complex^{m \times n}$ that allows us to conclude \eqref{equation:conclusionOfRIPInLevelsTypeTheorem} and \eqref{equation:conclusionOfRIPInLevelsTypeTheoremL2}. If the RIP implies \eqref{equation:conclusionOfRIPTypeTheorem} then the obvious idea is to extend the RIP to a so-called `RIP in levels', defined as follows:
\begin{definition}[RIP in levels]
\label{def:RIPInLevels}
For a given sparsity pattern $\smsparse$ and matrix $U \in \complex^{m \times n}$, the \emph{RIP in levels} ($\mathrm{RIP}_L$) constant of order $\smsparse$, denoted by $\delta_{\mathbf{s,M}}$, is the smallest $\delta> 0$ such that
\begin{equation*}
(1-\delta)\|x\|^2_2 \leq \|Ux\|^2_2 \leq (1+\delta) \|x\|^2_2
\end{equation*}
for all $x \in \Sigma_{\mathbf{s,M}}$.
\end{definition}

We will see that the RIP in levels allows us to obtain error estimates on $\|x-\widetilde{x}\|_1$ and $\|x-\widetilde{x}\|_2$.

\section{Main results}\label{Section:Results} 
If a matrix $U \in \complex^{m \times n}$ satisfies the RIP then we have control over the values of $\|Ue_i\|_2$ where $i \in \lbrace 1,2, \dotsc, n\rbrace$ and $e_i$ is the $i$-th standard basis element of $\complex^n$. To ensure that the same thing happens with the $\ripl$ we make the following two definitions:
\begin{definition}[Ratio constant]
\label{Definition:RatioConstant}
The \emph{ratio constant} of a sparsity pattern $\smsparse$, which we denote by $\eta_{\mathbf{s,M}}$, is given by
$\eta_{\mathbf{s,M}} := \max_{i,j} s_i/s_j.$

%\begin{equation*}
%\eta_{\mathbf{s,M}} := \max_{i,j} \frac{s_i}{s_j}.
%\end{equation*}
If the sparsity pattern $\smsparse$ has $l$ levels and there is a $j \in \lbrace 1,2,\dotsc,l \rbrace$ for which $s_j = 0$ then we write $\eta_{\mathbf{s,M}} = \infty$.
\end{definition}

\begin{definition}\label{def:CoverMatrix}
A sparsity pattern $\smsparse$ is said to \emph{cover} a matrix $U \in \complex^{m \times n}$ if 
\begin{enumerate}
\item $\etasm < \infty$ \label{Assumption:EtaFinite}
\item $M_{l} \geq n$ where $l$ is the number of levels for $\smsparse$. \label{Assumption:GoesToEndOfMatrix}
\end{enumerate}
\end{definition}

If a sparsity pattern does not cover $U$ because it fails to satisfy either \ref{Assumption:EtaFinite} or \ref{Assumption:GoesToEndOfMatrix} from the definition of a sparsity pattern covering a matrix $U$ then we cannot guarantee recovery of $\smsparse$-sparse vectors, even in the case that $\delta_{\mathbf{s,M}} = 0$. We shall justify the necessity of both conditions using two counterexamples. Firstly, we shall provide a matrix $U$, a sparsity pattern $\smsparse$ and an $\smsparse$-sparse vector $x_1 \in \complex^n$ such that $\etasm = \infty$, $\delta_{\mathbf{s,M}}  = 0$ and $x_1$ is not recovered by standard $\ell^1$ minimisation. Indeed, consider the following
\begin{equation*}
U = \begin{pmatrix}
1 & 2 \\
0 & 0 \\
\end{pmatrix}\text{, }\hspace{6pt}\mathbf{s} = (1,0) \text{, }\hspace{6pt} \mathbf{M} = (0,1,2)\text{, }\hspace{6pt} x_1 = \begin{pmatrix} 1\\ 0 \end{pmatrix}. 
\end{equation*}
By the definition of $\etasm$, we have $\etasm = \infty$ and it is obvious that $\delta_{\mathbf{s,M}} = 0$. Furthermore, even without noise, $x_1$ does not solve the minimisation problem $\min \|\tilde{x}\|_1 \text{ such that } Ux_1 = U\tilde{x}.$ This can easily be seen by observing that $Ux_1 = Ux_2$ with $\|x_2\|_1 = \frac{1}{2}$ where $x_2 := (0, 1/2)^{T}$.
It is therefore clear that Assumption \ref{Assumption:EtaFinite} is necessary. We shall now provide an explanation for why Assumption \ref{Assumption:GoesToEndOfMatrix} is also a requirement if we wish for the $\ripl$ to be a sufficient condition for the recovery of $\smsparse$-sparse vectors. This time, consider the following combination of $U$, $\smsparse$ and $x_1$:
\begin{equation*}
U = \begin{pmatrix}
1 & 0 & 2 \\
0 & 1 & 0 \\
\end{pmatrix}\text{, }\hspace{6pt}\mathbf{s} = (1) \text{, }\hspace{6pt} \mathbf{M} = (0,1)\text{, }\hspace{6pt} x_1 = (1, 0, 0)^T. 
\end{equation*}
and again, even though $\delta_{\mathbf{s,M}} = 0$, recovery is not possible because $Ux_1 = Ux_2$ with $\|x_2 \|_1 = 1/2$ where
$x_2:= (0, 0, 1/2)^T$.

We shall therefore try to prove a result similar to Theorem \ref{Theorem:RIPtypeTheorem} for the $\ripl$ under the assumption that $\smsparse$ covers $U$.  
We need one further definition to state a result equivalent to Theorem \ref{Theorem:RIPtypeTheorem} for the $\ripl$. In equation \eqref{equation:conclusionOfRIPTypeTheorem} the bound on $||x-\widetilde{x}||_1$ involves $\sqrt{s}$. This arises because $s$ is the maximum number of non-zero values that could be in an $s$-sparse vector. The equivalent for $\smsparse$-sparse vectors is the following:
\begin{definition}\label{def:NumElements}
The \emph{number of elements} of a sparsity pattern $\smsparse$, which we denote by $\widetilde{s}$, is given by $\tilde{s}:= s_1+s_2+\dotsb + s_l$.
\end{definition}
To prove that a sufficiently small RIP in levels constant implies an equation of the form \eqref{equation:conclusionOfRIPInLevelsTypeTheorem}, it is natural to adapt the steps used in  \cite{4over41Paper} to prove Theorem \ref{Theorem:RIPtypeTheorem}.
This adaptation yields a sufficient condition for recovery even in the noisy case.

\begin{theorem}\label{Theorem:RIPInLevelsRecoveryTheorem}
Let $\smsparse$ be a sparsity pattern with $l$ levels and ratio constant $\eta_{\mathbf{s,M}}$. 
Suppose that the matrix $U \in \mathbb{C}^{m \times n}$ is covered by $\smsparse$ and has a $\ripl$ constant $\delta_{2\mathbf{s,M}}$ satisfying
\begin{equation}\label{eq:RequirementOnDelta2sM}
\delta_{2\mathbf{s,M}} < \frac{1}{\sqrt{l \left(\sqrt{\eta_{\mathbf{s,M}}}+\frac{1}{4}\right)^2 + 1}}.
\end{equation}
Furthermore, suppose that $x \in \mathbb{C}^n$ and $y \in \mathbb{C}^m$ satisfy $\|Ux-y\|_2 \leq \epsilon$. 
Then any $\widetilde{x} \in \complex^n$ which solve the $\ell^1$ minimisation problem
\begin{equation*} 
\min_{\widehat{x} \in \complex^n} \|\widehat{x}\|_1 \text{ \emph{subject to} } \|U\widehat{x} - y\|_2 \leq \epsilon 
\end{equation*}
also satisfy
\begin{align} 
\|x-\widetilde{x}\|_1 &\leq C_1 \sigma_{\mathbf{s,M}}(x)_1 + D_1 \sqrt{\tilde{s}}\,\epsilon \,\,\,\text{and} \label{eq:RIPLL1EstimateImplication}\\ 
\|x - \widetilde{x}\|_2 &\leq   \frac{\sigma_{\mathbf{s,M}}(x)_1}{\sqrt{\tilde{s}}} \left(C_2+C'_2\sqrt[4]{l\eta_{\mathbf{s,M}} }\right) + \epsilon \left(D_2+D'_2 \sqrt[4]{l\eta_{\mathbf{s,M}}}\right)\label{eq:RIPLL2EstimateImplication}
\end{align}
where $C_1,C_2,C'_2,D_1,D_2$ and $D'_2$ depend only on $\delta_{2\mathbf{s,M}}$.
\end{theorem}

This result allows uniform recovery within the class of $\smsparse$-sparse vectors but the requirement on $\delta_{2\mathbf{s,M}}$ depends on $l$ and $\eta_{\mathbf{s,M}}$. We make the following observations:
\begin{enumerate}
\item If we pick a sparsity pattern that uses lots of levels then we will require a smaller $\ripl$ constant. 
\item If we pick a sparsity pattern with fewer levels then typically we shall pick a collection of $s_i$ so that $\frac{s_i}{s_j}$ is correspondingly larger for distinct $i$ and $j$. 
\item If the $\ripl$ constant $\delta_{2\mathbf{s,M}}$ is sufficiently small so that the conclusion of Theorem \ref{Theorem:RIPInLevelsRecoveryTheorem} holds, the bound on 
$\|x-x\|_2$ is weaker than the bound in Theorem \ref{Theorem:RIPtypeTheorem}. 
\end{enumerate}
As a consequence of these observations, at first glance it may appear that the results we have obtained with the $\ripl$ are weaker than those obtained using the standard RIP. However, Theorem \ref{Theorem:RIPInLevelsRecoveryTheorem} is stronger than Theorem \ref{Theorem:RIPtypeTheorem} in two senses. Firstly, if one considers a sparsity pattern with one level then Theorem \ref{Theorem:RIPInLevelsRecoveryTheorem} reduces to Theorem \ref{Theorem:RIPtypeTheorem}. Secondly, the conclusion of Theorem \ref{Theorem:RIPtypeTheorem} does not apply at all if we do not have the RIP. Therefore, for the examples given in Figure \ref{fig:flipTestImage}, Theorem \ref{Theorem:RIPtypeTheorem} does not apply at all. 

Ideally, it would be possible to find a constant $C$ such that if the $\ripl$ constant is smaller than $C$ then recovery of all $\smsparse$-sparse vectors would be possible. Unfortunately, we shall demonstrate that this is impossible in Theorems \ref{Theorem:etaDependenceTheorem} and \ref{Theorem:lDependenceTheorem}. Indeed, in some sense Theorem \ref{Theorem:RIPInLevelsRecoveryTheorem} is optimal in $l$ and $\etasm$, as the following results confirm.

\begin{theorem}\label{Theorem:etaDependenceTheorem}
Fix $a \in \mathbb{N}$ and $f: \real \to \real$ such that $f(\eta_{\mathbf{s,M}}) = o(\eta_{\mathbf{s,M}}^{\frac{1}{2}})$. Then there are $m,n \in \mathbb{N}$, a matrix $U\in \complex^{m \times n}$ and a sparsity pattern $\smsparse$ with two levels that covers $U$ such that the $\ripl$ constant $\delta_{a\mathbf{s,M}}$ and ratio constant $\eta_{\mathbf{s,M}}$ satisfy
\begin{equation}
\delta_{a\mathbf{s,M}} \leq \frac{1}{|f(\eta_{\mathbf{s,M}})|} \label{equation:etaDependenceRIPBounds}
\end{equation}
but there is an $\smsparse$-sparse $z^1$ such that
\begin{equation*}
z^1 \notin \arg\min \|z\|_1 \emph{\text{ such that }} Uz = Uz^1.
\end{equation*}
\end{theorem}
Roughly speaking, Theorem \ref{Theorem:etaDependenceTheorem} says that if we fix the number of levels and try to replace the condition
\begin{equation*}
\delta_{2\mathbf{s,M}} < \frac{1}{\sqrt{l \left(\sqrt{\eta_{\mathbf{s,M}}}+\frac{1}{4}\right)^2 + 1}}
\end{equation*}
with a condition of the form $\delta_{2\mathbf{s,M}} < \left(\eta_{\mathbf{s,M}}\right)^{-\alpha/2}/(C\sqrt{l})$ for some constant $C$ and some $\alpha < 1$ then the conclusion of Theorem \ref{Theorem:RIPInLevelsRecoveryTheorem} ceases to hold. In particular, the requirement on $\delta_{2\mathbf{s,M}}$ cannot be independent of $\etasm$. The parameter $a$ in the statement of Theorem \ref{Theorem:etaDependenceTheorem} says that we cannot simply fix the issue by changing $\delta_{2\mathbf{s,M}}$ to $\delta_{3\mathbf{s,M}}$ or any further multiple of $\mathbf{s}$.

Similarly, we can state and prove a similar theorem that shows that the dependence on the number of levels, $l$, cannot be ignored.
\begin{theorem}\label{Theorem:lDependenceTheorem}
Fix $a \in \mathbb{N}$ and $f: \real \to \real$ such that $f(l) = o(l^{\frac{1}{2}})$. Then there are $m,n \in \mathbb{N}$, a matrix $U \in \complex^{m \times n}$ and a sparsity pattern $\smsparse$ that covers $U$ with ratio constant $\eta_{\mathbf{s,M}} = 1$ and $l$ levels such that the $\ripl$ constant $\delta_{a\mathbf{s,M}}$ corresponding to $U$ satisfies
$\delta_{a\mathbf{s,M}} \leq 1/|f(l)|$
but there is an $\smsparse$-sparse $z^1$ such that
\begin{equation*}
z^1 \notin \arg\min \|z\|_1 \emph{\text{ such that }} Uz = Uz^1.
\end{equation*}
\end{theorem}
Furthermore, Theorem \ref{Theorem:RIPInLevelsOptimalL2} shows that the $\ell^2$ error estimate on $\|x-\widetilde{x}\|_2$ is optimal up to constant terms. 

\begin{theorem}\label{Theorem:RIPInLevelsOptimalL2}
The $\ell^2$ result (\ref{eq:RIPLL2EstimateImplication}) in Theorem \ref{Theorem:RIPInLevelsRecoveryTheorem} is sharp in the following sense:
\begin{enumerate}
\item For a fixed $a \in \mathbb{N}$ and any functions $f,g: \real \to \real$ such that $f(\eta) = o(\eta^{\frac{1}{4}})$ and $g(\eta) = O(\sqrt{\eta})$, there are natural numbers $m$ and $n$, a matrix $U_2 \in \complex^{m \times n}$ and a sparsity pattern $\smsparse$ with two levels that such that
\begin{itemize}
\item $\smsparse$ covers $U_2$
\item The $\ripl$ constant corresponding to the sparsity pattern $(a\mathbf{s,M})$, denoted by $\delta_{a\mathbf{s,M}}$, satisfies $\delta_{a\mathbf{s,M}} \leq 1/|g(\eta_{\mathbf{s,M}})|$.
\item There exist vectors $z$ and $z^1$ such that $U_2(z-z^1)=0$ and $\|z\|_1 \leq \|z^1\|_1$ but 
\begin{equation*} 
\|z-z^1\|_2 > \frac{f(\eta_{\mathbf{s,M}})}{\sqrt{\tilde{s}}} \sigma_{\mathbf{s,M}}(z^1)_1.
\end{equation*} 
\end{itemize}
\item For a fixed $a \in \mathbb{N}$ and any functions $f,g: \real \to \real$ such that $f(l) = o(l^{\frac{1}{4}})$ and $g(l) = O(\sqrt{l})$, there are natural numbers $m$ and $n$, a matrix $U_2 \in \complex^{m \times n}$ and a sparsity pattern $\smsparse$ with $\eta_{\mathbf{s,M}}=1$ such that 
\begin{itemize}
\item $\smsparse$ covers $U_2$
\item The $\ripl$ constant corresponding to the sparsity pattern $(a\mathbf{s,M})$, denoted by $\delta_{a\mathbf{s,M}}$, satisfies
$\delta_{a\mathbf{s,M}} \leq 1/|g(l)|.$
\item There exist vectors $z$ and $z^1$ such that $U_2(z-z^1)=0$ and $\|z\|_1 \leq \|z^1\|_1$ but 
\begin{equation*} 
\|z-z^1\|_2 > \frac{f(l)}{\sqrt{\tilde{s}}} \sigma_{\mathbf{s,M}}(z^1)_1.
\end{equation*} 
\end{itemize}
\end{enumerate}
\end{theorem}

Theorems of a similar form to Theorem \ref{Theorem:RIPtypeTheorem} are typically proven by showing that the \emph{$\ell^2$ robust nullspace property of order $s$} holds.
\begin{definition}\label{Definition:L2NSPOfOrderS}
A matrix $U \in \complex^{m \times n}$ is said to satisfy the $\ell^2$ robust nullspace property of order $s$ if there is a $\rho \in (0,1)$ and a $\tau>0$ such that for all vectors $v \in \complex^n$ and all $S$ which are subsets of $\lbrace 1,2,3, \dotsc,n \rbrace$ with $|S| \leq  s$, we have $
\|v_S\|_2 \leq \rho \|v_{S^c}\|_1/\sqrt{s}  + \tau \|Uv\|_2.$
\end{definition}

The implication is then the following Theorem: (for example, see \cite{foucartBook}, Theorem 4.22)

\begin{theorem}\label{Theorem:RobustL2NSPImplicationOrderS}
Suppose that $U \in \complex^{m \times n}$ satisfies the $\ell^2$ robust nullspace property of order $s$ with constants $\rho \in (0,1)$ and $\tau>0$.
Then any solutions $\widetilde{x} \in \complex^n$ to the $\ell^1$ minimisation problem 
\begin{equation*} 
\min_{\widehat{x} \in \complex^n} \|\widehat{x}\|_1 \text{ \emph{subject to} } \|U\widehat{x} - y\|_2 \leq \epsilon 
\end{equation*}
where 
$
\|Ux - y\|_2   \leq \epsilon
$ 
 satisfy
\begin{equation*}
\|\widetilde{x}-x\|_1 \leq C_1\sigma_{s}(x)_1 + D_1\epsilon \sqrt{s}, \quad \|\widetilde{x}-x\|_2 \leq \frac{C_2\sigma_{s }(x)_1}{\sqrt{s}} + D_2\epsilon 
\end{equation*}
where the constants $C_1,C_2,D_1$ and $D_2$ depend only on $\rho$ and $\tau$.
\end{theorem}

The corresponding natural extension of the $\ell^2$ robust nullspace of order $s$ to the $\smsparse$-sparse case is the \emph{$\ell^2$ robust nullspace property of order $\smsparse$}.
\begin{definition}\label{def:l2NSPROBUSTSM}
A matrix $U \in \complex^{m \times n}$ satisfies the $\ell^2$ robust nullspace property of order $\smsparse$ if there is a $\rho \in (0,1)$ and a $\tau>0$ such that
\begin{equation}\label{eq:NSPDefinitionOrderSM}
\|v_S\|_2 \leq \frac{\rho}{\sqrt{\tilde{s}}} \|v_{S^c}\|_1 + \tau \|Uv\|_2
\end{equation}
for all $\smsparse$-sparse sets $S$ and vectors $v \in \complex^n$.
\end{definition} 

The $\smsparse$-sparse version of Theorem \ref{Theorem:RobustL2NSPImplicationOrderS} is Theorem \ref{Theorem:RobustL2NSPImplicationOrderSM}.

\begin{theorem}\label{Theorem:RobustL2NSPImplicationOrderSM}
Suppose that a matrix $U \in \complex^{m \times n}$ satisfies the $\ell^2$ robust nullspace property of order $\smsparse$ with constants $\rho \in (0,1)$ and $\tau>0$. Let $x \in \mathbb{C}^n$ and $y \in \mathbb{C}^m$ satisfy $\|Ux-y\|_2 < \epsilon$.
Then any solutions $\widetilde{x}$ of the $\ell^1$ minimisation problem
\begin{equation*} 
\min_{\widehat{x} \in \complex^n} \|\widehat{x}\|_1 \text{ \emph{subject to} } \|U\widehat{x} - y\|_2 \leq \epsilon 
\end{equation*}
satisfy
\begin{align}
\|\widetilde{x}-x\|_1 &\leq A_1\sigma_{\mathbf{s,M}}(x)_1 + C_1 \epsilon \sqrt{\tilde{s}} \label{eq:RNSPL1EstimateImplication}\\
\|\widetilde{x}-x\|_2 &\leq \frac{\sigma_{\mathbf{s,M}}(x)_1}{\sqrt{\tilde{s}}}\left(A_2 + B_2\sqrt[4]{l\eta_{\mathbf{s,M}}}\right) + 2\epsilon \left(C_2+D_2 \sqrt[4]{l\eta_{\mathbf{s,M}}}\right) \label{eq:RNSPL2EstimateImplication}
\end{align}
where 
$$
A_1:= \frac{2 + 2\rho}{1-\rho}, \quad C_1:= \frac{4\tau}{1-\rho}, \quad A_2 := \frac{2\rho+2\rho^2}{1-\rho},
$$
$$
B_2:= \frac{\left(2\sqrt{\rho}+1\right)\left(1+\rho\right)}{1-\rho}, \quad C_2 := \frac{\rho\tau+\tau}{1-\rho} \ \text{ and } \  D_2:= \frac{4\sqrt{\rho}\tau + 3\tau - \rho\tau}{2-2\rho}. 
$$
\end{theorem}

This Theorem explains where the dependence on $\etasm$ and $l$ in \eqref{eq:RIPLL2EstimateImplication} emerges from. Analogous to proving the $\ell^2$ error estimates in Theorem \ref{Theorem:RIPtypeTheorem} using the $\ell^2$ robust nullspace property of order $s$, we prove the $\ell^2$ error estimate in Theorem \ref{Theorem:RIPInLevelsRecoveryTheorem} by showing that a sufficiently small $\ripl$ constant implies the robust $\ell^2$ nullspace property of order $\smsparse$. The $\ell^2$ error estimate \eqref{eq:RNSPL2EstimateImplication} follows and we are left with a dependence on $\sqrt[4]{l\eta_{\mathbf{s,M}}}$ in the right hand side of \eqref{eq:RIPLL2EstimateImplication}.
As before, Theorem \ref{Theorem:RNSPOptimalL2} shows that this is optimal.

\begin{theorem}\label{Theorem:RNSPOptimalL2}
The result in Theorem \ref{Theorem:RobustL2NSPImplicationOrderSM} is sharp, in the sense that
\begin{enumerate}
\item 
For any $f:\real^3 \to \real$ satisfying $f(\rho,\tau,\eta) = o(\eta^{\frac{1}{4}}) \text{ for fixed } \rho \in (0,1) \text{ and } \tau > 0,$
there are natural numbers $m$ and $n$, a matrix $U_2 \in \complex^{m \times n}$ and a sparsity pattern $\smsparse$ with ratio constant $\eta_{\mathbf{s,M}}$ and two levels such that 
\begin{itemize}
\item $\smsparse$ covers $U_2$
\item $U_2$ satisfies the $\ell^2$ robust nullspace property of order $\smsparse$ with constants $\rho \in (0,1)$ and $\tau > 0$
\item There exist vectors $z$ and $z^1$ such that $U_2(z-z^1)=0$ and $\|z\|_1 \leq \|z^1\|_1$ but 
\begin{equation*} 
\|z-z^1\|_2 > \frac{f(\rho,\tau,\etasm)}{\sqrt{\tilde{s}}} \sigma_{\mathbf{s,M}}(z^1)_1.
\end{equation*} 
\end{itemize}
\item 
For any $f:\real^3 \to \real$ satisfying $f(\rho,\tau,l) = o(l^{\frac{1}{4}}) \text{ for fixed } \rho \in (0,1) \text{ and } \tau > 0,$
there are natural numbers $m$ and $n$, a matrix $U_2 \in \complex^{m \times n}$ and a sparsity pattern $\smsparse$ with ratio constant $\eta_{\mathbf{s,M}} = 1$ and $l$ levels such that
\begin{itemize}
\item $\smsparse$ covers $U_2$
\item $U_2$ satisfies the $\ell^2$ robust nullspace property of order $\smsparse$ with constants $\rho \in (0,1)$ and $\tau > 0$
\item There exist vectors $z$ and $z^1$ such that $U_2(z-z^1)=0$ and $\|z\|_1 \leq \|z^1\|_1$ but 
\begin{equation*} 
\|z-z^1\|_2 > \frac{f(\rho,\tau,l)}{\sqrt{\tilde{s}}} \sigma_{\mathbf{s,M}}(z^1)_1.
\end{equation*}
\end{itemize}

\end{enumerate}
\end{theorem}
The conclusions that we can draw from the above theorems are the following:
\begin{enumerate}
\item The $\ripl$ will guarantee $\ell^1$ and $\ell^2$ estimates on the size of the error $\|\widetilde{x}-x\|$, provided that the $\ripl$ constant is sufficiently small (Theorem \ref{Theorem:RIPInLevelsRecoveryTheorem}). 
\item The requirement that the $\ripl$ constant is sufficiently small is dependent on $\sqrt{\eta_{\mathbf{s,M}}}$ and $\sqrt{l}$. This is optimal up to constants (Theorem \ref{Theorem:etaDependenceTheorem} and Theorem \ref{Theorem:lDependenceTheorem}).
\item The $\ell^2$ error when using the $\ripl$ has additional factors of the form $\sqrt[4]{l}$ and $\sqrt[4]{\eta_{\mathbf{s,M}}}$. Again, these are optimal up to constants (Theorem \ref{Theorem:RIPInLevelsOptimalL2}).
\item Typically, we use the robust $\ell^2$ nullspace property of order $s$ to give us a bound on the $\ell^2$ error in using $\widetilde{x}$ to approximate $x$. When this concept is extended to a robust $\ell^2$ nullspace property of order $\smsparse$ then the $\ell^2$ error gains additional factors of the form $\sqrt[4]{l}$ and $\sqrt[4]{\eta_{\mathbf{s,M}}}$ (Theorem \ref{Theorem:RobustL2NSPImplicationOrderSM}).
\item These factors are optimal up to constants, so that even if we ignore the $\ripl$ and still try to prove results using the $\ell^2$ robust nullspace property of order $\smsparse$ then we would be unable to improve the $\ell^2$ error (Theorem \ref{Theorem:RNSPOptimalL2}).
\end{enumerate}
With these results, we have demonstrated that the RIP in levels may be able to explain why permutations within levels are possible and why more general permutations are impossible. The results that we have obtained give a sufficient condition on the RIP in levels constant that guarantees $\smsparse$-sparse recovery. Furthermore, we have managed to demonstrate that this condition and the conclusions that follow from it are optimal up to constants.

\section{Conclusions and open problems}
The flip test demonstrates that in practical applications the ability to recover sparse signals depends on the structure of the sparsity, so that a tool that guarantees uniform recovery of all $s$-sparse signals does not apply.
 The flip test with permutations within the levels suggests that reasonable sampling schemes provide a different form of uniform recovery, namely, the recovery of  $\smsparse$-sparse signals. It is therefore natural to try to find theoretical tools that are able to analyse and describe this phenomenon.
However, we are now left with the fundamental problems:
\begin{itemize}
\item Given a sampling basis (say Fourier) and a recovery basis (say wavelets) and a sparsity pattern $\smsparse$, what kind of sampling procedure will give the $\mathrm{RIP}_L$ of order $\smsparse$?
\item How many samples must one draw, and how should they be distributed (as a function of $\mathbf{s}$ and $\mathbf{M}$), in order to obtain the $\mathrm{RIP}_L$?
\end{itemize}

Note that these problems are vast as the sampling patterns will not only depend on the sparsity patterns, but of course also on the sampling basis and recovery basis (or frame). Thus, covering all interesting cases relevant to practical applications will yield an incredibly rich mathematical theory.

\section*{Acknowledgements}
The authors would like to thank Arash Amini for asking an excellent question, during a talk given at EPFL, about the possibility of having a theory for the RIP in levels. It was this question that sparked the research leading to this paper. The authors would also like to thank Ben Adcock for many great suggestions and input. Finally, the authors would like to thank Bogdan Roman for his valuable discussions concering the implementation of the examples using the hadamard transformation. All the numerical computations were done with the SPGL1 package \cite{SPGL1,SPGL1Paper}.

ACH acknowledges support from a Royal Society University Research Fellowship as well as the UK Engineering and Physical Sciences Research Council (EPSRC) grant EP/L003457/1 and AB acknowledges support from EPSRC grant EP/H023348/1 for the University of Cambridge Centre for Doctoral Training, the Cambridge Centre for Analysis.

\section{Proofs}
We shall present the proofs in a different arrangement to the order in which their statements were presented. The first proof that we shall present is that of Theorem \ref{Theorem:RobustL2NSPImplicationOrderSM}.
\subsection{Proof of Theorem \ref{Theorem:RobustL2NSPImplicationOrderSM}}
We begin with the following lemma: 

\begin{lemma}\label{Lemma:NSP2ImplicationForV}
Suppose that $U \in \complex^{m \times n}$ satisfies the $\ell^2$ robust nullspace property of order $\smsparse$ with constants $\rho \in (0,1)$ and $\tau >0$. Fix $v \in \complex^n$, and let $S$ be an $\smsparse$-sparse set such that $|S| = \tilde{s}$ and the property that if $T$ is an $\smsparse$-sparse set, we have $\|v_S\|_1 \geq \|v_T\|_1$. Then
\begin{equation*}
\|v\|_2 \leq \frac{\|v_{S^c}\|_1}{\sqrt{\tilde{s}}} \left[ \rho + \left(\sqrt{\rho} + \frac{1}{2}\right) \sqrt[4]{l \eta_{\mathbf{s,M}}} \right]+ \tau \|Uv\|_2 \left[ \frac{\sqrt[4]{l\eta_{\mathbf{s,M}}}}{2} + 1 \right].
\end{equation*}
\end{lemma}
\begin{proof}
For $i=1,2, \dotsc, l$, we define $S^i_0$ to be 
$
S^i_0 = S \cap \left\lbrace M_{i-1}+1,M_{i-1}+2, \dotsc, M_i \right\rbrace
$
(i.e. $S^i_0$ is the elements of $S$ that are in the $i$th level). Let 
$
m = \max_{i=1,2, \dotsc, l} \,\min_{j \in S^i_0}  |v_j|. 
$ 
Since $|S^i_0| = s_i$ (otherwise $|S| < \tilde{s}$), we can see that given any $i=1,2, \dotsc, l$
\begin{align*}
\|v_S\|_2  &= \sqrt{\sum_{n \in S} |v_n|^2}  \geq \sqrt{\sum_{j \in S^i_0} |v_j|^2} \geq \sqrt{s_i} \min_{j \in S^i_0}  |v_j| \geq \min_{k=1,2, \dotsc, l} \!\sqrt{s_k}\; \min_{j \in S^i_0}  |v_j|
\end{align*}
so that 
$
\|v_S\|_2 \geq m\min \limits_{k=1,2, \dotsc, l} \sqrt{s_k}.
$ 
Furthermore, 
$ |v_j| \leq m$ for each $j \in S^c$
otherwise there is an $\smsparse$-sparse $T$ with $\|v_T\|_1 > \|v_S\|_1$. Therefore
%\begin{equation*}
$\|v_{S^c}\|^2_2 = \sum_{j \in S^c} |v_j|^2 \leq \sum_{j \in S^c} m|v_j|   \leq \frac{\|v_{S^c}\|_1 \|v_S\|_2}{\min\limits_{k =1,2,\dotsc,l}\sqrt{s_k}}.$
%\end{equation*} 
By the $\ell^2$ robust nullspace property of order $\smsparse$,
%\begin{equation*}
$\|v_{S^c}\|_1 \|v_S\|_2 \leq \frac{\rho}{\sqrt{\tilde{s}}}\|v_{S^c}\|_1^2 + \tau \|Uv\|_2 \|v_{S^c}\|_1.$
%\end{equation*}
Since $\sqrt{a+b} \leq \sqrt{a} + \sqrt{b}$ whenever $a,b>0$, 
\begin{equation}\label{eq:vSCLowerBoundNoSquare}
\|v_{S^c}\|_2 \leq \frac{1}{\min\sqrt[4]{s_i}} \left( \frac{\sqrt{\rho}}{\sqrt[4]{\tilde{s}}}\|v_{S^c}\|_1 + \sqrt{\tau \|Uv\|_2 \|v_{S^c}\|_1}\right)
\end{equation}
Using the arithmetic-geometric mean inequality,
\begin{align*}
\sqrt{\tau \|Uv\|_2 \|v_{S^c}\|_1} & = \sqrt{\tau \|Uv\|_2 \sqrt[4]{\tilde{s}} \frac{\|v_{S^c}\|_1}{\sqrt[4]{\tilde{s}}}} \leq \frac{\tau \|Uv\|_2 \sqrt[4]{\tilde{s}}}{2} +  \frac{\|v_{S^c}\|_1}{2\sqrt[4]{\tilde{s}}}
\end{align*}
Therefore, \eqref{eq:vSCLowerBoundNoSquare} yields
\begin{align*}
\|v_{S^c}\|_2 &\leq \frac{1}{\min\sqrt[4]{s_i}} \left(  \frac{\sqrt{\rho}}{\sqrt[4]{\tilde{s}}}\|v_{S^c}\|_1 + \frac{\|v_{S^c}\|_1}{2\sqrt[4]{\tilde{s}}} + \frac{\tau \|Uv\|_2 \sqrt[4]{\tilde{s}}}{2} \right)\\
&\leq \frac{\|v_{S^c}\|_1}{\sqrt[4]{\tilde{s}}\min\sqrt[4]{s_i}}\left(\sqrt{\rho} + \frac{1}{2} \right) + \frac{\tau \|Uv\|_2 \sqrt[4]{l\eta_{\mathbf{s,M}}}}{2}
\end{align*}
because 
$
\frac{\tilde{s}}{\min s_i} \leq l \eta_{\mathbf{s,M}}.
$
Once again, employing the $\ell^2$ nullspace property gives
\begin{align*}
\|v\|_2 \leq \|v_S\|_2 + \|v_{S^c}\|_2 &\leq \frac{\rho}{\sqrt{\tilde{s}}} \|v_{S^c}\|_1 + \tau \|Uv\|_2 + \frac{\|v_{S^c}\|_1}{\sqrt[4]{\tilde{s}}\min\sqrt[4]{s_i}}\left(\sqrt{\rho} + \frac{1}{2} \right) + \frac{\tau \|Uv\|_2 \sqrt[4]{l\eta_{\mathbf{s,M}}}}{2}\\
&\leq \frac{\|v_{S^c}\|_1}{\sqrt{\tilde{s}}} \left[ \rho + \left(\sqrt{\rho} + \frac{1}{2}\right)\frac{\sqrt[4]{
\tilde{s}}}{\min\sqrt[4]{s_i}}\right] + \tau \|Uv\|_2 \left[\frac{\sqrt[4]{l\eta_{\mathbf{s,M}}}}{2} + 1\right]\\
&\leq \frac{\|v_{S^c}\|_1}{\sqrt{\tilde{s}}} \left[ \rho + \left(\sqrt{\rho} + \frac{1}{2}\right)\sqrt[4]{l\eta_{\mathbf{s,M}}}\right] + \tau \|Uv\|_2 \left[\frac{\sqrt[4]{l\eta_{\mathbf{s,M}}}}{2} + 1\right].
\end{align*}
\end{proof}
The remaining error estimates will follow from various properties related to the \emph{$\ell^1$ robust nullspace property} (see \cite{foucartBook}, definition 4.17) holds. To be precise, 
\begin{definition}\label{def:l1robustNSP}
A matrix $U \in \complex^{m \times n}$ satisfies the $\ell^1$ robust nullspace property relative to $S$ with constants $\rho \in (0,1)$ and $\tau' > 0$ if
\begin{equation}\label{eq:l1NSPDefinition}
\|v_S\|_1 \leq \rho \|v_{S^c}\|_1 + \tau'  \|Uv\|_2 
\end{equation}
for any $v \in \complex^n$. We say that $U$ satisfies the $\ell^1$ robust nullspace property of order $\smsparse$ if \eqref{eq:l1NSPDefinition} holds for any $\smsparse$-sparse sets $S$. 
\end{definition}
It is easy to see that if $U$ satisfies the $\ell^2$ robust nullspace property of order $\smsparse$ with constants $\rho$ and $\tau$ then, for any $\smsparse$-sparse set $S$, $U$ also satisfies the $\ell^1$ robust nullspace property relative to $S$ with constants $\rho$ and $\tau\sqrt{\tilde{s}}$. Indeed, assume that $U$ satisfies the $\ell^2$ robust nullspace property of order $\smsparse$ with constants $\rho$ and $\tau$. Then (by the Cauchy-Schwarz inequality) 
%\begin{equation*}
$\|v_S\|_1 \leq \sqrt{\tilde{s}}\|v_S\|_2 \leq \rho \|v_{S^c}\|_1 + \tau \sqrt{\tilde{s}} \|Uv\|_2.$
%\end{equation*}

An immediate conclusion of the robust nullspace property is the following, proven in \cite{foucartBook} as Theorem 4.20.
\begin{lemma}\label{lemmaNSPRelativetoSConclusion}
Suppose that $U \in \complex^{m \times n}$ satisfies the $\ell^1$ robust null space property with constants $\rho \in (0,1)$ and $\tau'$ relative to a set $S$. Then for any complex vectors $x,z \in \complex^n$, we have
\begin{equation*}
\|z-x\|_1 \leq \frac{1+\rho}{1-\rho} \left( \|z\|_1 - \|x\|_1 + 2 \|x_{S^c}\|_1 \right) + \frac{2\tau'}{1-\rho} \|U(z-x)\|_2.
\end{equation*}
\end{lemma}

We can use this lemma to show the following important result, which is similar both in proof and statement to Theorem 4.19 in \cite{foucartBook}.
\begin{lemma}\label{Lemma:NSPSufficiency}
Suppose that a matrix $U \in \complex^{m \times n}$ satisfies the $\ell^1$ robust nullspace property of order $\smsparse$ with constants $\rho \in (0,1)$ and $\tau'>0$. Furthermore, suppose that
$
\|Ux - y\|_2 \leq \epsilon.
$
 Then any solutions $\widetilde{x}$ to the $\ell^1$ minimisation problem
\begin{equation*} 
\min_{\widehat{x} \in \complex^n} \|\widehat{x}\|_1 \text{ \emph{subject to} } \|U\widehat{x} - y\|_2 \leq \epsilon 
\end{equation*} 
satisfy
\begin{equation*} 
\|x-\widetilde{x}\|_1 \leq  \frac{2+2\rho}{1-\rho} \sigma_{\mathbf{s,M}}(x)_1 + \frac{4\tau'\epsilon }{1-\rho}.
\end{equation*}
\end{lemma}
\begin{proof}
By Lemma \ref{lemmaNSPRelativetoSConclusion}, for any $\smsparse$-sparse set $S$
\begin{equation*}
\|\widetilde{x}-x\|_1 \leq \frac{1+\rho}{1-\rho} \left(\|\widetilde{x}\|_1 - \|x\|_1 + 2 \|x_{S^c}\|_1\right)+ \frac{2\tau'}{1-\rho} \|U(\widetilde{x}-x)\|_2
\end{equation*}
Because both $\|Ux-y\|_2$ and $\|U\widetilde{x} - y\|_2$ are smaller than or equal to $\epsilon$,
$
\|Ux-U\widetilde{x}\| \leq 2\epsilon.
$
Furthermore, because $\widetilde{x}$ has minimal $\ell_1$ norm, 
$
\|\widetilde{x}\|_1 - \|x\|_1 \leq 0.
$

Thus
$
\|x-\widetilde{x}\|_1 \leq \frac{2+2\rho}{1-\rho}  \|x_{S^c}\|_1+ \frac{4\tau'\epsilon}{1-\rho}.
$
If we take $S$ to be the $\smsparse$-sparse set which maximizes $\|x_S\|_1$, then
\begin{equation*}
\|x-\widetilde{x}\|_1 \leq \frac{2+2\rho}{1-\rho}  \sigma_{\mathbf{s,M}}(x)_1 + \frac{4\tau' \epsilon}{1-\rho}.
\end{equation*}
\end{proof}
We can combine these results to complete the proof of Theorem \ref{Theorem:RobustL2NSPImplicationOrderSM}. Indeed, \eqref{eq:RNSPL1EstimateImplication} follows immediately from Lemma \ref{Lemma:NSPSufficiency} and the fact that $U$ satisfies the $\ell^1$ robust nullspace property with constants $\rho$ and $\tau \sqrt{\tilde{s}}$. To prove \eqref{eq:RNSPL2EstimateImplication}, we can simply set $v = x-\widetilde{x}$ in Lemma \ref{Lemma:NSP2ImplicationForV} to see that
\begin{align*}
\|x-\widetilde{x}\|_2 &\leq \frac{\|(x-\widetilde{x})_{S^c}\|_1}{\sqrt{\tilde{s}}} \left[ \rho + \left(\sqrt{\rho} + \frac{1}{2}\right) \sqrt[4]{l \eta_{\mathbf{s,M}}} \right]+ \tau \|U\left(x-\widetilde{x}\right)\|_2 \left[ \frac{\sqrt[4]{l\eta_{\mathbf{s,M}}}}{2} + 1 \right]\\
&\leq \frac{\|x-\widetilde{x}\|_1}{\sqrt{\tilde{s}}} \left[ \rho + \left(\sqrt{\rho} + \frac{1}{2}\right) \sqrt[4]{l \eta_{\mathbf{s,M}}} \right]+ 2\tau\epsilon \left[ \frac{\sqrt[4]{l\eta_{\mathbf{s,M}}}}{2} + 1 \right]
\end{align*}
and the result follows from $\eqref{eq:RNSPL1EstimateImplication}$.

\subsection{Proof of Theorem \ref{Theorem:RIPInLevelsRecoveryTheorem}}
It will suffice to prove that the conditions on $\delta_{\mathbf{s,M}}$ and $\smsparse$ in Theorem \ref{Theorem:RIPInLevelsRecoveryTheorem} imply the $\ell^2$ robust nullspace property. To show this, we begin by stating the following inequality, proven in \cite{NormInequality}:
\begin{lemma}[The norm inequality for $\ell_1$ and $\ell_2$]\label{l1l2NormInequality}
Let $v = (v_1,v_2, \dotsc ,v_s)$ where $v_1 \geq v_2 \geq v_3 \geq \dotsb  \geq v_s$. Then 
\begin{equation*}
\|v\|_2 \leq \frac{1}{\sqrt{s}} \|v\|_1 + \frac{\sqrt{s}}{4} (v_1 - v_s)
\end{equation*}
\end{lemma}

We will now prove the following additional lemma which is almost identical in statement and proof to that of Lemma 6.1 in \cite{4over41Paper}.	
\begin{lemma}\label{lemmaInnerProductBetweenVectors}
Suppose that $x,y \in \Sigma_{\mathbf{s,M}}$ and that
\begin{equation} \label{equation:InnerProductBetweenVectorsAssumption1}
\|Ux\|^2_2 - \|x\|^2_2 = t\|x\|^2_2.
\end{equation}
Additionally, suppose that $x$ and $y$ are orthogonal.
Then 
$
|\langle Ux,Uy \rangle| \leq \sqrt{\delta_{2\mathbf{s,M}}^2-t^2}\|x\|_2\|y\|_2 
$
where $\delta_{2\mathbf{s,M}}$ is the restricted isometry constant corresponding to the sparsity pattern $(2\mathbf{s,M})$ and the matrix $U$.
\end{lemma}
\begin{proof}
Without loss of generality, we can assume that $\|x\|_2 = \|y\|_2 = 1$. Note that for $\alpha,\beta \in \real$ and $\gamma \in \mathbb{C}$, the vectors $\alpha x + \gamma y$ and $\beta x - \gamma y$ are contained in $\Sigma_{2\mathbf{s,M}}$. Therefore,
\begin{equation}\label{equation:firstVecUB}
\|U\left(\alpha x + \gamma y\right)\|^2_2 \leq (1+\delta_{2\mathbf{s,M}}) \|\alpha x + \gamma y\|^2_2 = (1+\delta_{2\mathbf{s,M}})(\alpha^2+|\gamma|^2).
\end{equation}
where the last line follows because $\left\langle x,y \right\rangle = 0$ (from the orthogonality of $x$ and $y$).
Similarly,
\begin{equation}\label{equation:secondVecUB}
-\|U(\beta x - \gamma y)\|^2_2 \leq -(1-\delta_{2\mathbf{s,M}})  (\beta^2+|\gamma|^2)
\end{equation}
We will now add these two inequalities. On the one hand (by using the assumption in \eqref{equation:InnerProductBetweenVectorsAssumption1} and the fact that $\alpha$,$\beta$ are real), we have
\begin{align*}
\|U(\alpha x+ \gamma y)\|^2_2  - \|U(\beta x - \gamma y)\|^2_2 &= \alpha^2 \|Ux\|^2_2 + 2\rp(\alpha\overline{\gamma}\left\langle Ux,Uy \right\rangle) + |\gamma|^2 \|Uy\|^2_2 \\
&\text{ } -\left(\beta^2 \|Ux\|^2_2 - 2\rp(\beta\overline{\gamma} \left\langle Ux,Uy\right\rangle) + |\gamma|^2 \|Uy\|^2_2\right)\\
&=(1+t) \left(\alpha^2 - \beta^2\right) + 2(\alpha+\beta)\rp(\overline{\gamma} \left\langle Ux,Uy \right\rangle )
\end{align*}
and on the other hand (from \eqref{equation:firstVecUB} and \eqref{equation:secondVecUB})
\begin{equation*}
\|U(\alpha x+ \gamma y)\|^2_2  - \|U(\beta x - \gamma y)\|^2_2 \leq \delta_{2\mathbf{s,M}} \left(\alpha^2+\beta^2 + 2|\gamma|^2 \right) + \alpha^2 - \beta^2.
\end{equation*}
Therefore,
\begin{equation*}
(1+t) \left(\alpha^2 - \beta^2\right) + 2(\alpha+\beta)\rp(\overline{\gamma} \left\langle Ux,Uy \right\rangle )  \leq \delta_{2\mathbf{s,M}} \left(\alpha^2+\beta^2 + 2|\gamma|^2 \right) + \alpha^2 - \beta^2.
\end{equation*}
After choosing $\gamma$ so that $\rp(\overline{\gamma} \left\langle Ux,Uy \right\rangle) = |\left\langle Ux,Uy \right\rangle|$ we obtain
\begin{equation}\label{equation:UxUyUpperBound}
|\left\langle Ux,Uy \right\rangle| \leq \frac{1}{2\alpha+2\beta}\left[(\delta_{2\mathbf{s,M}}-t)\alpha^2 + ( \delta_{2\mathbf{s,M}}+t)\beta^2 + 2\delta_{2\mathbf{s,M}}\right]	
\end{equation}  
because $|\gamma|=1$. 
By the definition of the RIP in levels constant, $\delta_{2\mathbf{s,M}} \geq \delta_{\mathbf{s,M}}$ and so
\begin{equation}
|t| = \left| \|Ux\|_2^2 - \|x\|^2_2 \right| \leq \delta_{\mathbf{s,M}} \leq \delta_{2\mathbf{s,M}}. \label{eq:tDelta2sM}
\end{equation}
If equality holds in \eqref{eq:tDelta2sM}, then we can set $\beta = 0$ and send $\alpha \to \infty$ in \eqref{equation:UxUyUpperBound} to obtain the required result. Otherwise, equation \eqref{eq:tDelta2sM} implies that 
$
 \sqrt{\frac{\delta_{2\mathbf{s,M}} +t}{\delta_{2\mathbf{s,M}} - t}} \in \real
$
and so we can set
$
\alpha = \sqrt{\frac{\delta_{2\mathbf{s,M}} +t}{\delta_{2\mathbf{s,M}} - t}} 
$
and $\beta = \frac{1}{\alpha}$
in equation \eqref{equation:UxUyUpperBound}.
With these values, we obtain
\begin{align*}
|\left\langle Ux,Uy \right\rangle| &\leq  \frac{\alpha}{2\alpha^2+2} \left( \delta_{2\mathbf{s,M}} + t + \delta_{2\mathbf{s,M}} - t + 2\delta_{2\mathbf{s,M}}\right) \leq \frac{4\delta_{2\mathbf{s,M}}\alpha(\delta_{2\mathbf{s,M}}-t)}{4\delta_{2\mathbf{s,M}}}  \leq 	\sqrt{\delta^2_{2\mathbf{s,M}}-t^2}. 
\end{align*}
\end{proof}

\begin{proof}[Proof (of Theorem \ref{Theorem:RIPInLevelsRecoveryTheorem})]
Let $x \in \complex^m$ be an arbitrary $m$ dimensional complex vector, and let
\begin{equation*}
x^i  := x_{\left\lbrace M_{i-1}+1,M_{i-1}+2, \dotsc ,M_{i} \right\rbrace}
\end{equation*}
denote the $i$th level of $x$. For an arbitrary vector $v = \left(v_1,v_2, \dotsc ,v_n\right)$, we define $|v|$ to be the vector $\left(|v_1|,|v_2|, \dotsc,|v_n|\right)$. 
Let $S^i_0$ denote the indexes of the $s_i$th largest elements of $|x^i|$, and 
$
S_0 := \bigcup_{i=1}^{l}S^i_0.
$
We then define $S^i_1$ to be the indexes of the $s_i$th largest elements of $|x^i|$ that are not contained in $S^i_0$ (if there are fewer that $s_i$ elements remaining, we simply take the indexes of any remaining elements of $|x^i|$) and define
$
S_1 := \bigcup_{i=1}^{l}S^i_1.
$
In general, we can make a similar definition to form a collection of index sets labelled $(S^i_j)_{i=1,2 \dotsc,l,j=1,2, \dots}$ and corresponding $\smsparse$-sparse $S_j$. 

These definitions and the fact that $\smsparse$ covers $U$ implies that if $\Omega = \bigcup\limits_{j \geq 0}S_j$ then $x_{\Omega} = x$.
By the definition of $S_0$,
$
\|x_{\Lambda}\|_1 \leq \|x_{S_0}\|_1
$
whenever $\Lambda$ is $\smsparse$-sparse. By Theorem \ref{Theorem:RobustL2NSPImplicationOrderSM}, it will suffice to verify that
\begin{equation}
\sqrt{\tilde{s}} \|x_{S_0}\|_2 \leq \rho \|x_{S^c_0}\|_1 + \tau \sqrt{\tilde{s}}\|Ux\|_2 \label{requiredNSP}
\end{equation}
holds for some $\rho \in (0,1)$ and $\tau>0$. Set 
\begin{equation}
\|Ux_{S_0}\|_2^2 = (1+t)\|x_{S_0}\|_2^2. \label{initialdef}
\end{equation}
Clearly, $|t| \leq \delta_{\mathbf{s,M}}$. Then
\begin{align}
\|Ux_{S_0}\|^2_2 &= \left\langle Ux_{S_0},Ux_{S_0} \right\rangle  = \left\langle Ux_{S_0},Ux \right\rangle - \sum_{j \geq 1} \left\langle Ux_{S_0},Ux_{S_j}\right\rangle. \label{simpleupperbound}
\end{align} 
where we have used $x_{\Omega} = x$.
Using the Cauchy-Schwarz inequality and \eqref{initialdef} yields
\begin{equation}
|\left\langle Ux_{S_0},Ux \right\rangle| \leq \|Ux_{S_0}\|_2 \|Ux\|_2 \leq \sqrt{1+t}\,\|x_{S_0}\|_2 \|Ux\|_2. \label{noisetermub}
\end{equation}
Furthermore, we can use Lemma \ref{lemmaInnerProductBetweenVectors} to see that
\begin{align}
\left|\,\sum_{j \geq 1} \left\langle Ux_{S_0},Ux_{S_j} \right\rangle\right| &\leq \sqrt{\delta^2_{2\mathbf{s,M}}-t^2}\sum_{j \geq 1} \|x_{S_0}\|_2\|x_{S_j}\|_2 \leq \|x_{S_0}\|_2 \sqrt{\delta^2_{2\mathbf{s,M}}-t^2}\sum_{i=1}^l\sum_{j \geq 1} \|x_{S^i_j}\|_2. \label{initialupperbound}
\end{align}
Combining \eqref{initialdef},\eqref{simpleupperbound},\eqref{noisetermub} and \eqref{initialupperbound} yields 
\begin{equation}
(1+t)\|x_{S_0}\|_2^2 \leq \sqrt{1+t}\,\|x_{S_0}\|_2  \|Ux\|_2 +\|x_{S_0}\|_2  \sqrt{\delta^2_{2\mathbf{s,M}}-t^2} \sum_{i=1}^{l}\sum_{j \geq 1} \|x_{S^i_j}\|_2. \label{eq:ProofOf4.4Interim}
\end{equation}
If $|S^i_j| = s_i$ then let $x^{+}_{i,j}$ (correspondingly $x^{-}_{i,j}$) be the largest element of $\left|x_{S^i_j}\right|$ (correspondingly the smallest element of $\left|x_{S^i_j}\right|$). If $S^i_j$ is non-empty with fewer than $s_i$ elements then we set $x^{+}_{i,j}$ to be the largest element of $\left|x_{S^i_j}\right|$ and $x^{-}_{i,j}=0$. Finally, when $S^i_j = \emptyset$, we let $x^{+}_{i,j}=x^{-}_{i,j}=0$. It is clear then that $x^{+}_{i,j+1} \leq x^{-}_{i,j}$.

Since $x_{S^i_j}$ contains at most $s_i$ non-zero elements, we can apply the norm inequality for $\ell_1$ and $\ell_2$ (Lemma \ref{l1l2NormInequality}) to obtain
\begin{equation*}
\|x_{S^i_j}\|_2 \leq \frac{1}{\sqrt{s_i}}\|x_{S^i_j}\|_1 + \frac{\sqrt{s_i}}{4} \left(x^+_{i,j} - x^-_{i,j}\right)
\end{equation*}
for any $i=1,2,\dotsc,l$ and $j \in \mathbb{N}$. Therefore
\begin{align*}
\sum_{j \geq 1} \|x_{S^i_j}\|_2 &\leq \sum_{j \geq 1} \left(\frac{1}{\sqrt{s_i}}\|x_{S^i_j}\|_1\right) + \frac{\sqrt{s_i}}{4} \sum_{j \geq 1}\left(x^+_{i,j} - x^-_{i,j}\right)\\
& \leq \sum_{j \geq 1} \left(\frac{1}{\sqrt{s_i}}\|x_{S^i_j}\|_1\right) + \frac{\sqrt{s_i}}{4} \left(x^+_{i,1} + \sum_{j \geq 2} x^+_{i,j} - \sum_{j \geq 1} x^+_{i,j}\right)\\
& \leq \sum_{j \geq 1} \left(\frac{1}{\sqrt{s_i}}\|x_{S^i_j}\|_1\right) + \frac{\sqrt{s_i}}{4} \left(x^+_{i,1} + \sum_{j \geq 1}\left(x^+_{i,j+1} -x^-_{i,j}  \right)\right)\\
& \leq \sum_{j \geq 1} \left(\frac{1}{\sqrt{s_i}}\|x_{S^i_j}\|_1\right) + \frac{\sqrt{s_i}}{4} x^+_{i,1}
\end{align*}
where the last inequality follows because $x^+_{i,j+1} - x^-_{i,j} \leq 0$. Additionally, 
\begin{align*}
\frac{\sqrt{s_i}}{4} x^{+}_{i,1} &= \frac{1}{4}\sqrt{\underbrace{(x^{+}_{i,1})^2+(x^{+}_{i,1})^2+(x^{+}_{i,1})^2 + \dotsb +(x^{+}_{i,1})^2}_{s_i \text{ elements}}}\leq \frac{1}{4}\|x_{S^i_0}\|_2
\end{align*}
because each element of $\left|x_{S^i_0}\right|$ is larger than $x^{+}_{i,1}$. We conclude that
\begin{align*}
\sum_{i=1}^{l}\sum_{j \geq 1} \|x_{S^i_j}\|_2 &\leq \sum_{j \geq 1} \sum_{i=1}^{l} \frac{1}{\sqrt{s_i}} \|x_{S^i_j}\|_1 + \sum_{i=1}^{l} \frac{1}{4} \|x_{S^i_0}\|_2 \leq \frac{1}{\min\sqrt{s_i}}\sum_{j \geq 1} \sum_{i=1}^{l}  \|x_{S^i_j}\|_1 + \frac{1}{4}\sqrt{l}\|x_{S_0}\|_2 \nonumber \\
&\leq \frac{1}{\min\sqrt{s_i}} 	\sum_{j \geq 1} \|x_{S_j}\|_1 + \frac{1}{4}\sqrt{l}\|x_{S_0}\|_2 \leq  \frac{1}{\min\sqrt{s_i}} \left\|x_{ \bigcup\limits_{j\geq 1}S_j} \right\|_1 + \frac{1}{4} \sqrt{l} \|x_{S_0}\|_2 \nonumber \\
\end{align*}
where the second inequality follows from the Cauchy-Schwarz inequality applied to $(\underbrace{1,1, \dotsc ,1}_{l})$ and $(\|x_{S^1_0}\|_2,\|x_{S^2_0}\|_2, \dotsc ,\|x_{S^l_0}\|_2)$ and the third and fourth inequalities follow from the disjoint supports of the vectors $x_{S^i_j}$ and $x_{S^{i'}_{j'}}$ whenever $i \neq i'$ or $j \neq j'$.
By $x_{\Omega} = x$ and the disjointedness of $S_i,S_j$ for $i \neq j$, 
$
\bigcup\limits_{j\geq 1}S_j = S_0^c
$
so
\begin{equation} \label{upperboundxvecS^i_j}
\sum_{i=1}^{l}\sum_{j \geq 1} \|x_{S^i_j}\|_2 \leq   \frac{1}{\min\sqrt{s_i}} \|x_{S_0^c}\|_1 + \frac{1}{4} \sqrt{l} \|x_{S_0}\|_2.
\end{equation}
Dividing \eqref{eq:ProofOf4.4Interim} by $\|x_{S_0}\|_2$ and employing \eqref{upperboundxvecS^i_j} yields
\begin{equation}
(1+t)\|x_{S_0}\|_2 \leq \sqrt{1+t} \|Ux\|_2 +  \sqrt{\delta^2_{2\mathbf{s,M}}-t^2} \left(\frac{1}{\min\sqrt{s_i}} \|x_{S_0^c}\|_1 + \frac{1}{4} \sqrt{l} \|x_{S_0}\|_2\right) . \label{eqAlmostThere}
\end{equation}
Let 
$
g(t):= \frac{\delta^2_{2\mathbf{s,M}}-t^2}{(1+t)^2}
$
for $|t| \leq \delta_{2\mathbf{s,M}}$. It is clear that $g(\delta_{2\mathbf{s,M}}) = g(-\delta_{2\mathbf{s,M}}) = 0$. Furthermore, $g$ is differentiable. Therefore $g$ attains its maximum at $t_{\max}$, where $g'(t_{\max}) = 0$.
A simple calculation shows us that $t_{\max} = -\delta_{2\mathbf{s,M}}^2$ (note that by the assumption \eqref{eq:RequirementOnDelta2sM}, $\delta_{2\mathbf{s,M}}^2 \leq \delta_{2\mathbf{s,M}}$). Thus
$g(t) \leq g(-\delta^2_{2\mathbf{s,M}}) = \frac{\delta_{2\mathbf{s,M}}^2}{1-\delta_{2\mathbf{s,M}}^2}.$ Additionally,
$
\frac{1}{\sqrt{1+t}} \leq \frac{1}{\sqrt{1-\delta_{2\mathbf{s,M}}}}.
$
Combining this with \eqref{eqAlmostThere} yields
\begin{align*}
\|x_{S_0}\|_2 &\leq \frac{1}{\sqrt{1+t}}\|Ux\|_2 + \sqrt{g(t)}\left(\frac{1}{\min\sqrt{s_i}} \|x_{S_0^c}\|_1 + \frac{1}{4} \sqrt{l} \|x_{S_0}\|_2\right) \\
&\leq \frac{1}{\sqrt{1-\delta_{2\mathbf{s,M}}}} \|Ux\|_2 + \frac{\delta_{2\mathbf{s,M}}}{\sqrt{1-\delta_{2\mathbf{s,M}}^2}} \left(\frac{1}{\min\sqrt{s_i}} \|x_{S_0^c}\|_1 + \frac{1}{4} \sqrt{l} \|x_{S_0}\|_2\right).
\end{align*}
A simple rearrangement gives
\begin{equation}
\|x_{S_0}\|_2 \leq \frac{\sqrt{1+\delta_{2\mathbf{s,M}}}}{\sqrt{1-\delta_{2\mathbf{s,M}}^2}-\delta_{2\mathbf{s,M}}\sqrt{l}/4} \|Ux\|_2 + \frac{\delta_{2\mathbf{s,M}}}{\min{\sqrt{s_i}}\left(\sqrt{1-\delta_{2\mathbf{s,M}}^2}-\delta_{2\mathbf{s,M}}\sqrt{l}/4\right)}\|x_{S^c_0}\|_1 \label{boundwithoutsqrts}
\end{equation}
provided
\begin{equation}
\sqrt{1-\delta_{2\mathbf{s,M}}^2}-\delta_{2\mathbf{s,M}}\sqrt{l}/4 > 0. \label{firstcondition}
\end{equation}
Multiplying \eqref{boundwithoutsqrts} by $\sqrt{\tilde{s}}$ yields
\begin{align*}
\sqrt{\tilde{s}}\|x_{S_0}\| &\leq \sqrt{\tilde{s}} \frac{\sqrt{1+\delta_{2\mathbf{s,M}}}}{\sqrt{1-\delta_{2\mathbf{s,M}}^2}-\delta_{2\mathbf{s,M}}\sqrt{l}/4} \|Ux\|_2 + \frac{\delta_{2\mathbf{s,M}}\sqrt{\tilde{s}}}{\min{\sqrt{s_i}}\left(\sqrt{1-\delta_{2\mathbf{s,M}}^2}-\delta_{2\mathbf{s,M}}\sqrt{l}/4\right)}\|x_{S^c_0}\|_1 \\
& \leq \tau \sqrt{\tilde{s}} \,\|Ux\|_2 + \frac{\delta_{2\mathbf{s,M}}}{\sqrt{1-\delta_{2\mathbf{s,M}}^2}-\delta_{2\mathbf{s,M}}\sqrt{l}/4}\sqrt{\sum\limits_{k=1}^l\frac{ s_k}{\min s_i}}\|x_{S^c_0}\|_1 \\
& \leq \tau \sqrt{\tilde{s}}\, \|Ux\|_2 +  \frac{\delta_{2\mathbf{s,M}}\sqrt{l \etasm}}{\sqrt{1-\delta_{2\mathbf{s,M}}^2}-\delta_{2\mathbf{s,M}}\sqrt{l}/4}\|x_{S^c_0}\|_1
\end{align*}
where $\tau = \frac{\sqrt{1+\delta_{2\mathbf{s,M}}}}{\sqrt{1-\delta_{2\mathbf{s,M}}^2}-\delta_{2\mathbf{s,M}}\sqrt{l}/4}.$
It is clear that \eqref{requiredNSP} is satisfied if condition \eqref{firstcondition} holds and
\begin{equation}
 \frac{\delta_{2\mathbf{s,M}}\sqrt{l \etasm}}{\sqrt{1-\delta_{2\mathbf{s,M}}^2}-\delta_{2\mathbf{s,M}}\sqrt{l}/4} <1  \quad\text{ or equivalently } \, \delta_{2\mathbf{s,M}} < \frac{1}{\sqrt{l \left(\sqrt{\etasm}+\frac{1}{4}\right)^2 + 1}} \label{finalcondition} 
\end{equation}
whilst \eqref{firstcondition} is equivalent to 
$
\delta_{2\mathbf{s,M}} < \frac{1}{\sqrt{ \frac{l}{16}+1}}.
$
Since 
\begin{equation*}
 \frac{1}{\sqrt{l \left(\sqrt{\etasm}+\frac{1}{4}\right)^2 + 1}} \leq  \frac{1}{\sqrt{ \frac{l}{16}+1}}
\end{equation*}
it will suffice for \eqref{finalcondition} to hold, completing the proof.
\end{proof}

\subsection{Proof of Theorem \ref{Theorem:etaDependenceTheorem} and \ref{Theorem:lDependenceTheorem}}\label{Section:etaDependenceProof}

\begin{proof}[Proof of Theorem \ref{Theorem:etaDependenceTheorem}] The ideas behind the counterexample in this proof are similar to those in \cite{SharpRIPBounds}.
We prove this theorem in three stages. First we shall construct the matrix $U$. Next we shall show that our construction does indeed have a RIP in levels constant satisfying \eqref{equation:etaDependenceRIPBounds}. Finally, we shall explain why $z^1$ exists.

{\bf Step I:}
Set $n = C+C^2$, where the non-negative integer $C$ is much greater than $a$ (we shall give a precise choice of $C$ later). 
Let $x^1 \in \complex^n$ be the vector 
\begin{equation*}
x^1:=\lambda(\underbrace{C,C,\dotsc, C}_{C},\underbrace{1,1,\dotsc,1}_{C^2}).
\end{equation*}
With this definition, the first $C$ elements of $x^1$ have value $C \lambda$ and the next $C^2$ elements have value $\lambda$. Our $\smsparse$ sparsity pattern is given by $\mathbf{s}=(1,C^2)$ and $\mathbf{M}=(0,C,C+C^2)$. 
Clearly, by the definition of the ratio constant, $\eta_{\mathbf{s,M}} = C^2$ (in particular, $\eta_{\mathbf{s,M}}$ is finite). Choose
%\begin{equation}\label{sizeOfLambda}
$\lambda = \frac{1}{\sqrt{C^3+C^2}}$
%\end{equation} 
so that $\|x^1\|_2 = 1$.
By using this fact, we can form an orthonormal basis of $\complex^{C+C^2}$ that includes $x^1$. We can write this basis as $(x^i)_{i=1}^{C+C^2}$.
Finally, for a vector $v \in \complex^{C+C^2}$, we define the linear map $U$ by
\begin{equation*}
Uv := \sum_{i=2}^{C+C^2} v^i x^i \text{ where } v = \sum_{i=1}^{C+C^2} v^i x^i
\end{equation*}  
In particular, notice that the nullspace of $U$ is precisely the space spanned by $x^1$, and that 
$
v^i = \left\langle v,x^i \right\rangle.
$

{\bf Step II:}
Let $\mathbf{\gamma}$ be an $(a\mathbf{s,M})$ sparse vector.
Our aim will be to estimate $\left|\,\|U\gamma\|^2_2 - \|\gamma\|^2_2\,\right|$. Clearly,
$
\|U\gamma\|^2_2 - \|\gamma\|^2_2  = -|\gamma^1|^2,
$
where $\gamma^1$ is the coefficient of $x^1$ in the expansion of $\gamma$ in the basis $(x^i)$. 
Therefore, to show that $U$ satisfies the $\ripl$ we will only need to bound
$
|\gamma^1|=|\!\left\langle \gamma,x^1 \right\rangle\!|.
$
Let $S$ be the support of $\gamma$. Then
\begin{equation*}
|\!\left\langle \gamma,x^1 \right\rangle \!|=|\!\left\langle \gamma_{S},x^1 \right\rangle \!| = |\!\left \langle \gamma,x^1_{S} \right\rangle\!| \leq \|\gamma\|_2 \|x^1_{S}\|_2\leq  \lambda\|\gamma\|_2 \sqrt{aC^2+C^2}
\end{equation*}
where we have used Cauchy-Schwarz in the first inequality and in the second inequality we have used the fact that $x^1_S$ has at most $a$ elements of size $\lambda C$ and at most $C^2$ elements of size $\lambda$.
From the definition of $\lambda$ we get
$
\left|\left\langle \gamma,x^1 \right\rangle\right| \leq \sqrt{\frac{a+1}{C+1}}\|\gamma\|_2. 
$
Therefore,
$$
\left|\, \|U\gamma\|^2_2 - \|\gamma \|^2_2 \,\right| =|\!\left\langle\gamma,x^1\right\rangle\!|^2 \leq \frac{a+1}{C+1} \|\gamma\|_2^2.
$$
By the assumption that $f(x) = o(x^\frac{1}{2})$, we can find a $C \in \mathbb{N}$ sufficiently large so that
$
\frac{a+1}{C+1} \leq \frac{1}{|f(C^2)|}.
$
Then
$
\delta_{a\mathbf{s,M}} \leq \frac{1}{|f(\eta_{\mathbf{s,M}})|}
$
as claimed. 

{\bf Step III:}
Let
\begin{align*}
z^1:=(C,\underbrace{0,0,\dotsc,0,0}_{C-1},\underbrace{1,1,\dotsc,1}_{C^2}), \qquad z^2:=(0,\underbrace{C,C, \dotsc, C,C}_{C-1},\underbrace{0,0, \dotsc ,0}_{C^2}).
\end{align*}
It is clear that $z^1$ is $\smsparse$-sparse.
Additionally,
$
\|z^1\|_1 =  C^2+C
$
and
$
\|-z^2\|_1 = \left(C-1\right)C=C^2-C.
$
Because $U(z^1+z^2) = U(x^1)/\lambda = 0,$
 we have
$
U(-z^2) = U(z^1).
$
Since the kernel of $U$ is of dimension $1$, the only vectors $z$ which satisfy
$
U(z) =  U(z^1)
$
are $z=z^1$ and $z=-z^2$. Moreover,
$
\|z^1 \|_1 > \|-z^2\|_1.
$
Consequently
\begin{equation*}
z^1 \notin \arg\min \|z\|_1 \text{ such that } Uz = Uz^1.
\end{equation*}
\end{proof} 

\begin{proof}[Proof of Theorem \ref{Theorem:lDependenceTheorem}]
The proof of this theorem is almost identical to that of Theorem \ref{Theorem:etaDependenceTheorem}, so we shall omit details here. Again, we set $x^1$ so that
\begin{equation*}
x^1:=\lambda(\underbrace{C,C,\dotsc, C}_{C},\underbrace{1,1,\dotsc,1}_{C^2})
\end{equation*}
where $C \gg a$. We choose $\lambda$ so that $\|x^1\|_2=1$.
In contrast to the proof of Theorem \ref{Theorem:etaDependenceTheorem}, we take 
$$
\mathbf{s}=(\underbrace{1,1,1,\dotsc,1}_{C^2+1}), \qquad \mathbf{M}=(0,C,C+1, \dotsc, C+C^2-1,C+C^2).
$$
This time, there are $C^2+1$ levels and the ratio constant $\eta_{\mathbf{s,M}}$ is equal to $1$. Once again, we produce an orthonormal basis of $\complex^{C+C^2}$ that includes $x^1$, which we label $(x^i)_{i=1}^{C+C^2}$ and we define the linear map $U$ by
\begin{equation*}
Uv := \sum_{i=2}^{C+C^2} v^i x^i \text{ where } v = \sum_{i=1}^{C+C^2} v^i x^i.
\end{equation*}  
The same argument as before proves that for any $(a\mathbf{s,M})$-sparse $\gamma$, 
\begin{equation*}
\left|\, \|U\gamma\|^2_2 - \|\gamma\|^2_2 \,\right| \leq \frac{a+1}{C+1} \|\gamma\|_2^2
\end{equation*}
and again, taking $C$ sufficiently large so that
$
\frac{a+1}{C+1} \leq \frac{1}{|f(C^2+1)|}
$
yields
$
\delta_{a\mathbf{s,M}} \leq \frac{1}{|f(l)|}.
$
The proof of the existence of $z^1$ is the identical to Step III in the proof of Theorem \ref{Theorem:etaDependenceTheorem}.
\end{proof}

\subsection{Proof of Theorem \ref{Theorem:RIPInLevelsOptimalL2}}
\begin{proof} 
Once again, we prove this theorem in three stages. First we shall construct the matrix $U_2$. Next, we shall show that the matrix $U_2$ has a sufficiently small $\ripl$ constant. Finally, we shall explain why both $z^1$ and $z$ exist.

{\bf Step I:}
Let $x^1$ be the vector 
\begin{equation*}
x^1:=\lambda(\underbrace{0,0,\dotsc, 0}_{C^2},\underbrace{1,1,\dotsc,1}_{\omega(\rho,C)+1})
\end{equation*}
where $\omega(\rho,C) = \text{ceil}(\frac{2C}{\rho})$ for a fixed $\rho \in (0,1)$ which we will specify later, $\text{ceil}(a)$ denotes the smallest integer greater than or equal to $a$ and $C$ is an integer greater than $1$. 
In other words, the first $C^2$ elements of $x^1$ have value $0$ and the next $\omega(\rho,C) + 1$ elements have value $\lambda$. We choose $\lambda$ so that $\|x^1\|_2=1$ and $C$ so that $C^2 > \omega(\rho,C)$ and choose our $\smsparse$ sparsity pattern so that $\mathbf{s}=(C^2,1)$ and $\mathbf{M}=(0,C^2,C^2+\omega(\rho,C)+1)$. 
By the definition of the ratio constant, $\eta_{\mathbf{s,M}} = C^2$ (in particular, $\eta_{\mathbf{s,M}}$ is finite). 
Because $\|x^1\|_2 = 1$, we can form an orthonormal basis of $\complex^{C^2+\omega(\rho,C)+1}$ that includes $x^1$ which we can write as $(x^i)_{i=1}^{C^2+\omega(\rho,C)+1}$.
Finally, for a vector $v \in \complex^{C^2+\omega(\rho,C)+1}$, we define the linear map $U_2$ by
\begin{equation*}
U_2v := \frac{\sqrt{2}w}{\tau} \text{ where } w = \!\!\!\!\!\sum_{i=2}^{C^2+\omega(\rho,C)+1} \! \! \!\!\! v^i x^i \text{ and } v = v^1 x^1 + w
\end{equation*}  
In particular, notice that the nullspace of $U_2$ is precisely the space spanned by $x^1$, and that 
$
v^i = \left\langle v,x^i \right\rangle.
$

{\bf Step II:}
Let $\gamma$ be an $\left(a\mathbf{s,M}\right)$ sparse vector. For the purposes of proving Theorem \ref{Theorem:RIPInLevelsOptimalL2}, it will suffice to take $\tau = \sqrt{2}$. Then
 \begin{equation*}
\|U_2\gamma\|^2_2 - \|\gamma\|^2_2  = -|\gamma^1|^2,
\end{equation*}
where $\gamma^1$ is the coefficient corresponding to $x^1$ in the expansion of $\gamma$ in the basis $(x^i)$. 
As in the proof of Theorem \ref{Theorem:etaDependenceTheorem},
$
|\gamma^1|=|\!\left\langle \gamma,x^1 \right\rangle\!|.
$ Let $S$ be the support of $\gamma$. Then
\begin{align*}
|\!\left\langle \gamma_{S},x^1 \right\rangle \!| = |\!\left \langle \gamma,x^1_{S} \right\rangle\!| \leq \|\gamma\|_2 \|x^1_{S}\|_2 \leq  \lambda\|\gamma\|_2 \sqrt{a}  
\end{align*}
where we have used Cauchy-Scharwz in the first inequality and in the second inequality we have used the fact that $x^1_S$ has at most $a$ elements of size $\lambda$.
It is easy to see that $\lambda = \frac{1}{\sqrt{\omega(\rho,C)+1}}$.
Therefore,
\begin{align*}
\left|\, \|U_2\gamma\|^2_2 - \|\gamma\|^2_2 \,\right| &=|\!\left\langle\gamma,x^1\right\rangle\!|^2 \leq \frac{a}{\omega(\rho,C)+1} \|\gamma\|_2^2 \leq \frac{\rho a}{2C}.
\end{align*}
because $\omega(\rho,C) \geq \frac{2C}{\rho}$. 
By the assumption that $g(\eta_{\mathbf{s,M}}) \leq \frac{1}{A} \sqrt{\eta_{\mathbf{s,M}}}$ for some $A >0$ and $\eta_{\mathbf{s,M}}$ sufficiently large, and the fact that $\eta_{\mathbf{s,M}} = C^2$, we must have
$
\frac{A}{C} \leq \frac{1}{g(\eta_{\mathbf{s,M}})}.
$
If we take $\rho$ sufficiently small and $C$ sufficiently large, then
\begin{equation*}
\delta_{a\mathbf{s,M}} < \frac{\rho a}{2C} \leq \frac{A}{C} \leq \frac{1}{g(\eta_{\mathbf{s,M}})}.
\end{equation*}
as claimed. 

{\bf Step III:}
Let $z^1:=x^1$ and set $z$ to be the $0$ vector in  $\complex^{C^2 + \omega(\rho,C) + 1}$. 
Because $x^1$ is in the kernel of $U_2$, $U_2(z-z^1) = 0$. Furthermore, it is obvious that $\|z\|_1 \leq \|z^1\|_1$.
Additionally,
$
\|z^1\|_2 = 1
$
and
\begin{align*}
\frac{\sigma_{\mathbf{s,M}}(z^1)_1}{\sqrt{\tilde{s}}}&= \lambda \frac{\omega(\rho,C)}{\sqrt{C^2+1}} \leq \frac{\omega(\rho,C)}{\sqrt{\omega(\rho,C)\left(C^2+1\right)}} \leq  \sqrt{\frac{2C+1}{\rho\left(C^2+1\right)}}
\leq \sqrt{\frac{3}{\rho \sqrt{\etasm}}}
\end{align*}
since $\eta_{\mathbf{s,M}} = C^2$ and $\omega(\rho,C) \leq 2C/\rho +1 \leq (2C+1)/\rho$. Because $f(\eta_{\mathbf{s,M}}) = o(\eta_{\mathbf{s,M}}^{\frac{1}{4}})$,
\begin{equation*}
\frac{ \sigma_{\mathbf{s,M}}(z^1)_1}{\sqrt{\tilde{s}}} f(\eta_{\mathbf{s,M}}) \to 0, \quad \eta_{\mathbf{s,M}} \to \infty.
\end{equation*}
The desired result follows by taking $\eta_{\mathbf{s,M}}$ sufficiently large so that
\begin{equation*}
\|z-z^1\|_2 = 1 > \frac{ \sigma_{\mathbf{s,M}}(z^1)_1}{\sqrt{\tilde{s}}} f(\eta_{\mathbf{s,M}}).
\end{equation*}
\end{proof}
\begin{proof}[Proof of part 2]
The proof of part $2$ follows with a few alterations to the previous case. We now use the sparsity pattern 
\begin{equation*}
\mathbf{s} = (\underbrace{1,1,1, \dotsc, 1}_{C^2},1) \quad \text{and} \quad
\mathbf{M} = (0,1,2, \dotsc, C^2, C^2 + \omega(\rho,C)+1).
\end{equation*}
In this case, $\eta_{\mathbf{s,M}} = 1$ and $l = C^2 + 1$. The result follows by simply employing the same matrix $U_2$ with this new sparsity pattern.
\end{proof}

\subsection{Proofs of Theorem \ref{Theorem:RNSPOptimalL2}}
The counterexample for Theorem \ref{Theorem:RNSPOptimalL2} is the same as the one used in the proof of Theorem \ref{Theorem:RIPInLevelsOptimalL2}. In that case, the matrix depended on three parameters: $C, \tau$ and $\rho$. We show that $U_2$ satisfies the $\ell^2$ robust nullspace property of order $\smsparse$ with parameters $\rho$ and $\tau$. The existence of $z^1$ and $z$ is identical to Step III in the proof of Theorem \ref{Theorem:RIPInLevelsOptimalL2}. 
\begin{proof}[Proof of part 1]
Firstly, if $T \subset S$ then for any $v \in \complex^{C^2+\omega(\rho,C)+1}$, we have
\begin{equation*}
\|v_T\|_2 \leq \|v_S\|_2 \text{ and } \frac{\rho}{\sqrt{\tilde{s}}} \|v_{T^c}\|_1 + \tau \|U_2v\|_2 \geq \frac{\rho}{\sqrt{\tilde{s}}} \|v_{S^c}\|_1 + \tau \|U_2v\|_2
\end{equation*}
so it will suffice to prove that $U_2$ satisfies \eqref{eq:NSPDefinitionOrderSM} for $\smsparse$-sparse sets $S$ with $|S|= \tilde{s}$. \\
As before, we set $U_2v:= \sqrt{2}w/\tau$ where $w$ is defined as in the proof of Theorem \ref{Theorem:RIPInLevelsOptimalL2}.
Let us consider  a set $S$ such that $|S| = \tilde{s}$. Because $w_S$ and $w_{S^c}$ have disjoint support, by the Cauchy-Schwarz inequality applied to the vectors $(1,1)$ and $\left(\|w_S\|_2,\|w_{S^c}\|_2\right)$ we get
$
\sqrt{2}\|w\|_2 \geq \|w_S\|_2 + \|w_{S^c}\|_2.
$
Therefore,
\begin{equation} \label{eq:UVwSC+wSL2Bound}
\tau\|U_2v\|_2 \geq  \sqrt{2}\|w\|_2 \geq \|w_S\|_2 + \|w_{S^c}\|_2.
\end{equation}
Furthermore, because $|S| = \tilde{s} \geq |S^c|$ (recall that $|S| = C^2+1$ and that $C$ was chosen so that $ C^2 > \omega(\rho,C) = |S^c|$) and $\rho \in (0,1)$
\begin{align}
\|w_{S^c}\|_2 \geq \frac{1}{\sqrt{|S^c|}}\|w_{S^c}\|_1 \geq \frac{1}{\sqrt{\tilde{s}}} \|w_{S^c}\|_1  \geq \frac{\rho}{\sqrt{\tilde{s}}} \|w_{S^c}\|_1
 \label{eq:wSCL2wSCL1Bound}
\end{align}
Combining \ref{eq:UVwSC+wSL2Bound} and \eqref{eq:wSCL2wSCL1Bound} gives
\begin{align}
\tau \|Uv\|_2 + \frac{\rho}{\sqrt{\tilde{s}}}\|v_{S^c}\|_1 \geq \|w_S\|_2 + \frac{\rho}{\sqrt{\tilde{s}}} \|w_{S^c}\|_1 + \frac{\rho}{\sqrt{\tilde{s}}} \|v_{S^c}\|_1 &\geq \|w_S\|_2 + \frac{\rho}{\sqrt{\tilde{s}}} \|v_{S^c} - w_{S^c}\|_1 \nonumber\\
&\geq \|w_S\|_2 + \frac{\rho}{\sqrt{\tilde{s}}} \|v^1 x^1_{S^c}\|_1 \label{eq:reductionTox^1SC}
\end{align}
We shall now aim to bound $\|v^1 x^1_S\|_2$ in terms of $\|v^1 x^1_{S^c}\|_1$. We have
\begin{equation}\label{eq:v^1x^1_S|v^1|LowerBound}
\|v^1 x^1_S\|_2  \leq  \lambda |v^1|  
\end{equation}
since at most one element of $x^1_S$ is non-zero and its value will be at most $\lambda$. Additionally, since each element of $x^1_{S^c}$ has value $\lambda$ and there are at least $\frac{2C}{\rho}$ of them
\begin{align*}
\rho \|v^1 x^1_{S^c}\|_1 &= \rho  |v^1| \|x^1_{S^c}\|_1 \geq \frac{2\lambda C}{\rho} \rho  |v^1|  \geq 2\lambda C|v^1|.
\end{align*}
Therefore,
\begin{align}
\frac{\rho}{\sqrt{\tilde{s}}} \|v^1 x^1_{S^c}\|_1 &\geq \frac{2\lambda C}{ \sqrt{C^2+1}} |v^1| \geq \lambda |v^1| \label{eq:l1robustNSPLowerBoundx^1}.
\end{align}
Using \eqref{eq:v^1x^1_S|v^1|LowerBound} and \eqref{eq:l1robustNSPLowerBoundx^1}, we have
$
\|v^1 x^1_S\|_2  \leq \frac{\rho}{\sqrt{\tilde{s}}} \|v^1 x^1_{S^c}\|_1.
$
We can conclude the proof that $U_2$ satisfies the $\ell^2$ robust nullspace property by combining this result with \eqref{eq:reductionTox^1SC} as follows:
\begin{align*}
\|v_S\|_2 &\leq \|v^1 x^1_S\|_2 + \|w_S\|_2 \leq \frac{\rho}{\sqrt{\tilde{s}}} \|v^1 x^1_{S^c}\|_1 + \|w_S\|_2 \leq \tau \|Uv\|_2 + \frac{\rho}{\sqrt{\tilde{s}}}\|v_{S^c}\|_1.
\end{align*}

\end{proof}
\begin{proof}[Proof of part 2]
The proof of part $2$ is identical. We simply adapt the sparsity pattern so that 
\begin{equation*}
\mathbf{s} = (\underbrace{1,1,1, \dotsc, 1}_{C^2},1) \text{ and }\mathbf{M} = (0,1,2, \dotsc, C^2, C^2 + \omega(\rho,C)+1).
\end{equation*}
We can apply the proceeding argument with this new sparsity pattern to obtain the required result. 
\end{proof}

\bibliographystyle{abbrv}
\bibliography{RIPLevelsBib3}

\end{document}